\documentclass[11pt,a4paper,onecolumn]{IEEEtran}

\usepackage{graphicx}
\usepackage{amsmath, amsthm, amssymb, amsfonts, cancel}
\usepackage{algorithm, algorithmic, ifsym, subfigure}

\newcommand\blfootnote[1]{%
  \begingroup
  \renewcommand\thefootnote{}\footnote{#1}%
  \addtocounter{footnote}{-1}%
  \endgroup
}

\def\bd{ \textbf{d} } 
\def\bu{ \textbf{u} } 
\def\bp{ \textbf{p} } 
\def\ba{ \textbf{a} } 
\def\bA{ {\mathbf{A}} }

\def\bD{ {\mathbf{D}} }

\def\bI{ {\mathbf{I}} }

\def\bP{ {\mathbf{P}} }

\def\bY{ {\mathbf{Y}} }

\def\bepsilon{ {\mathbf{\epsilon}} }

\def\nn{{ \parallel   }}
\def\RR{{ \mathbb{R}  }}
\def\PP{{ \mathbb{P}  }}
\def\EE{{ \mathbb{E}  }}
\def\NN{{ \mathbb{N}  }}

\def\diag{{ \text{diag}   }}

\def\bx{{ \mathbf{x}  }}
\def\by{{ \mathbf{y}  }}

\def\be{{ \mathbf{e}  }}
\def\bv{{ \mathbf{v}  }}

\newtheorem{theorem}{Theorem}
\newtheorem{lemma}{Lemma}

\newtheorem{corollary}{Corollary}
\newtheorem{proposition}{Proposition}
\newtheorem{assumption}{Assumption}

\begin{document}
\title{On Decentralized Estimation with Active Queries}
\author{Theodoros Tsiligkaridis *, \textit{Member, IEEE}, Brian M. Sadler, \textit{Fellow, IEEE}, Alfred O. Hero III, \textit{Fellow, IEEE}}

\maketitle

\begin{abstract}
We consider the problem of decentralized 20 questions with noise for multiple players/agents under the minimum entropy criterion in the setting of stochastic search over a parameter space, with application to target localization. We propose decentralized extensions of the active query-based stochastic search strategy that combines elements from the 20 questions approach and social learning.
We prove convergence to correct consensus on the value of the parameter. This framework provides a flexible and tractable mathematical model for decentralized parameter estimation systems based on active querying. 
We illustrate the effectiveness and robustness of the proposed decentralized collaborative 20 questions algorithm for random network topologies with information sharing.
\end{abstract}


\blfootnote{The research reported in this paper was supported in part by ARO grant W911NF-11-1-0391.

T. Tsiligkaridis was with the Department of Electrical Engineering and Computer Science, University of Michigan, Ann Arbor, MI 48109 USA. He is now with MIT Lincoln Laboratory, Lexington, MA 02421 USA (email: ttsili@ll.mit.edu).

B. M. Sadler is with the US Army Research Laboratory, Adelphi, MD 20783 USA (email: brian.m.sadler6.civ@mail.mil).

A. O. Hero is with the Department of Electrical Engineering and Computer Science, University of Michigan, Ann Arbor, MI 48109 USA (email: hero@umich.edu).
}

\section{Introduction} \label{sec:intro}
Consider a set of agents that try to estimate a parameter, e.g., estimate a target state or location, collectively. The agents are connected by an information sharing network and can periodically query their local neighbors about the state of the target. Each agent's query is a binary question about target state that is formulated from the agent's local information. This decentralized estimation problem is an extension of the centralized version of the collaborative 20 questions framework studied in \cite{Tsiligkaridis:2013}. Unlike \cite{Tsiligkaridis:2013}, where a global centralized controller jointly or sequentially formulates optimal queries about target state for all agents, in the decentralized problem each agent formulates his own query based on his local information. Thus the decentralized collaborative 20 questions problem is relevant to large scale collaborative target tracking applications where there is no centralized authority. Examples include: object tracking in camera networks \cite{Sznitman:2010}; road tracking from satellite remote sensing networks \cite{Geman:1996}; and wide area surveillance networks \cite{CastroNowak07}.


In this paper we assume that the agents' observations obey noisy query-response models where the queries are functions of agents' local information and successive queries are determined by a feedback control policy. Specifically, in the 20 questions-type model considered in this paper, the observation of each agent is coupled with the query region chosen by that agent, which is a function of its current local belief.

In the framework of \cite{Tsiligkaridis:2013} a controller sequentially selects a set of questions about the target state and uses the noisy responses of the agents to formulate the next set of questions. The questions were binary partitions of the target state space and the agents provided binary answers about the state of the target in the partition. The agents' noisy response models were assumed to be binary symmetric channels (BSC) with known cross-over probability, which could be different for each agent, e.g., to account for a mixture of human and cyber agents. Under general conditions, it was shown that the optimal, entropy-minimizing joint query policy is equivalent to a sequential query policy. Furthermore, this optimal query policy was shown to reduce to a simple bisection rule that equalizes the posterior probabilities, the global belief function, that the target state lies in one of the partition elements. 

In this paper we extend the collaborative 20 questions framework of \cite{Tsiligkaridis:2013} to the decentralized case. The proposed decentralized algorithm consists of two stages: 1) local belief update; and 2) local information sharing. In stage 1 each  agent implements the bisection query policy of \cite{Tsiligkaridis:2013} to update their local belief function. In stage 2 the local belief functions are averaged  over nearest neighborhoods in the information sharing network. The information sharing stage is implemented by a neighborhood averaging rule similar to the social learning model proposed by Jadbabaie, et al., \cite{Jadbabaie:2012}.

We analyze both qualitative and quantitative properties of the two stage decentralized collaborative 20 questions algorithm and establish conditions under which the agents converge to a consensus estimate of the true state. 

The mathematical analysis of algorithm convergence in this paper is inspired by that of Jadbabaie, et al., \cite{Jadbabaie:2012} for social learning. However, we emphasize that the analysis of \cite{Jadbabaie:2012} is not directly applicable and requires significant extension to cover the collaborative 20 questions framework we propose.  First, in the 20 questions framework the target space is continuous as contrasted with the discrete case studied in \cite{Jadbabaie:2012}. Second, the 20 questions framework generates controlled observations that are not independent identically distributed, as required for the analysis in \cite{Jadbabaie:2012}. In particular, the controlled observation model leads to structured time-varying observation densities.

We establish the following theoretical properties of the proposed decentralized collaborative 20 questions algorithm. The first property is proved for states of arbitrary dimension while the second two properties are only proven for scalar states. 
\begin{itemize}
	\item  A positive linear combination of the integrated agents' belief functions over arbitrary sets forms a martingale sequence over time. This fact is the starting point for establishing asymptotic consensus via the martingale convergence theorem.   
	\item The agents asymptotically achieve consensus in their beliefs about target state, i.e., their belief functions converge to the same limit as time progresses (Thm. 1). Thus all agents asymptotically become in agreement about uncertainty in target state.
	\item The agents achieve consensus in the state of the target, i.e., their belief functions asymptotically concentrate on the true target state (Thm. 2).
\end{itemize}

In addition to theoretical analysis of the convergence of the proposed decentralized 20 questions algorithm, numerical simulations of algorithm performance are provided, showing interesting convergence behavior that information sharing brings. For a class of irreducible random graphs, the simulations show that little information sharing improves target localization average and worst-case root mean-square-error (RMSE) significantly when compared to the case of no information sharing, thus improving network-wide estimation performance.  

Our work also differs from the works on 20 questions/active stochastic search of Jedynak, et al., \cite{Jedynak12}, Castro and Nowak \cite{CastroNowak07}, Waeber, et al., \cite{Waeber:2013}, and Tsiligkaridis, et al., \cite{Tsiligkaridis:2013} because we consider intermediate local belief sharing between agents after each local bisection and update. In addition, in contrast to previous work, in the proposed framework each agent incorporates the beliefs of its neighbors in a way that is agnostic of its neighbors' error probabilities. We finally remark that, as compared to \cite{Tsiligkaridis:2013}, the proof of convergence of the proposed algorithm is complicated by the fact that the entropy of the posterior distribution for each agent in the network is not generally monotonically decreasing as a function of iteration. The analysis of \cite{Jadbabaie:2012} does not apply to our model since we consider controlled observations, although we use a form of the social learning model of \cite{Jadbabaie:2012}. 

We remark that our work differs from the large literature on consensus, see Dimakis, et al., \cite{Dimakis:2010} for a survey of gossip algorithms for sensor networks in the context of estimation, source localization and compression. Most of the work on linear consensus focuses on deriving conditions on the connectivity of the network such that all agents converge to the average of a static collection of measurements, along with rate of convergence analysis and development of fast consensus-achieving algorithms. In \cite{Aysal:2009}, randomized gossip broadcast algorithms for consensus were proposed and conditions for reaching consensus on the average value of the initial node measurements were presented. The mean-square error of the randomized averaging procedure was also studied and shown to decay monotonically to a steady-state value. In \cite{Kar:2011}, gossip algorithms for linear parameter estimation were studied and it was shown that, under appropriate conditions on the network structure and observation models, the distributed estimator achieves the same performance as the best centralized linear estimator in terms of asymptotic variance. In contrast, we consider a dynamic set of measurements, as there is novel information (or innovations) at each iteration step, given in the form of binary responses to actively-designed queries based on local agent information. These responses, over time, concentrate the agents' posterior probability distributions on the true target state. We believe that this a significant novel result since consensus-plus-controlled innovation type algorithms have not been rigorously studied in the literature.

The focus of this paper is to obtain a decentralized extension of the centralized collaborative 20 questions problem of \cite{Tsiligkaridis:2013}, and not on extending the analysis of the social learning algorithm of \cite{Jadbabaie:2012}. The decentralized extension enjoys numerous applications in large-scale controlled sensor networks, where sensors spread over wide areas can collaborate in an active manner to localize a target in the presence of errors. Other applications may include extending active testing approaches in the decentralized setting for classification problems, for instance in vision, recommendation systems, and epidemic networks. The 20 questions paradigm is motivated by asking the correct type of questions in the correct order and is applicable to various other domains where computational effort and time are critical resources to manage.

The outline of this paper is as follows. Section \ref{sec:notation} introduces the notation. Section \ref{sec:prior_work} briefly reviews some related prior work. Section \ref{sec:algorithm} introduces the decentralized estimation algorithm and its convergence properties are studied in Section \ref{sec:convergence}. 
The simulations are presented in Section \ref{sec:simulations} followed by our conclusions in Section \ref{sec:conclusions}. The proofs of convergence are given in the appendix.

\section{Notation} \label{sec:notation}
We define $X^*$ the true parameter, the target state in the sequel, and its domain as the unit interval $\mathcal{X}=[0,1]$. Let $\mathcal{B}(\mathcal{X})$ be the set of all Borel-measurable subsets $B \subseteq \mathcal{X}$. Let $\mathcal{N}=\{1,\dots,M\}$ index the $M$ agents in an interaction network, denoted by the vertex set $\mathcal{N}$ and the directed edges joining agents are captured by $E$. The directed graph $G=(\mathcal{N},E)$ captures the possible interactions between agents. Define the (first order) neighborhood in $G$ of agent $i$ as $\mathcal{N}_i=\{j \in \mathcal{N}: (j,i)\in E\}$. Define the probability space $(\Omega, \mathcal F, \mathbb P)$ consisting of the sample space $\Omega$ generating the unknown state $X^*$ and the observations $\{y_{i,t}\}$ at nodes $1\leq i\leq M$ and at times $t=1, 2, \ldots$, an event space $\mathcal F$ and a probability measure $\mathbb P$. The expectation operator $\mathbb E$ is defined with respect to $\mathbb P$.

Define the belief of the $i$-th agent at time $t$ on $\mathcal{X}$ as the posterior density $p_{i,t}(x)$  of target state $x\in \mathcal X$ based on all of the information available to this agent at this time. Define the $M\times 1$ vector $\bp_t(x)=[p_{1,t}(x),\dots,p_{M,t}(x)]^T$ for each $x \in \mathcal{X}$. For any $B\in \mathcal{B}(\mathcal{X})$, define $\bP_t(B)$ as the vector of probabilities with $i$-th element equal to $\int_B p_{i,t}(x) dx$. The interaction matrix $\bA=\{a_{i,j}\}$ (as in \cite{Jadbabaie:2012}) is defined to be any matrix $\bA$ consisting of nonnegative entries where each row sums to 1. We define the query point/target estimate of the $i$-th agent as $\hat{X}_{i,t}$. The query point is the right boundary of the region $A_{i,t}=[0,\hat{X}_{i,t}]$. We let $F_{i,t}(a)=\PP_{i,t}([0,a])=\int_{0}^a p_{i,t}(x) dx$ denote the cumulative distribution function associated with the density $p_{i,t}(\cdot)$.

We assume that each agent $i$ constructs a query at time $t$ of the form ``does $X^*$ lie in the region $A_{i,t} \subset \mathcal{X}$?''. We indicate this query with the binary variable $Z_{i,t}=I(X^* \in A_{i,t})$ to which each  agent $i$ responds with a binary response $Y_{i,t+1}$, which is correct with probability $1-\epsilon_i$, and by assumption $\epsilon_i<1/2$. This model for the error is equivalent to a binary symmetric channel (BSC) with crossover probability $\epsilon_i$. The query region $A_{i,t}$ at time $t$ depends on the accumulated information up to time $t$ at agent $i$. Define the nested sequence of event spaces $\mathcal{F}_t$, $\mathcal{F}_{t-1} \subset \mathcal{F}_{t}$, for all $t\geq 0$, generated by the sequence of queries and responses. The queries $\{A_{i,t}:1\leq i \leq M\}_{t\geq 0}$ are measurable with respect to this filtration. The notation \textit{i.p.} denotes convergence in probability and \textit{a.s.} denotes almost-sure convergence.



\section{Prior Work} \label{sec:prior_work}

\subsection{20 Questions \& Stochastic Search}
The paper by Jedynak, et al., \cite{Jedynak12} formulates the single player 20 questions problem as follows. A controller queries a noisy oracle about whether or not the state of a target $X^*$ lies in a set $A_n \subset \RR$. Starting with a prior distribution on the target's state $p_0(\cdot)$, the objective in \cite{Jedynak12} is to minimize the expected entropy of the posterior distribution:
\begin{equation} \label{eq: min_entr_obj}
	\inf_\pi \EE^\pi\left[ H(p_N) \right]
\end{equation}
where $\pi=(\pi_0,\pi_1,\dots)$ denotes the controller's query policy and the entropy is the standard differential entropy \cite{CoverThomas}:
\begin{equation*}
	H(p) = -\int_{\mathcal{X}} p(x) \log p(x) dx.
\end{equation*}
The posterior median of $p_N$ is used to estimate the target state after $N$ questions. Jedynak \cite{Jedynak12} shows the bisection policy is optimal under the minimum entropy criterion. To be concrete, in Thm. 2 of \cite{Jedynak12}, optimal policies are characterized by:
\begin{equation} \label{eq: optimal_policy}
	\PP_n(A_n) := \int_{A_n} p_n(x) dx = u^* \in \arg \max_{u \in [0,1]} \phi(u)
\end{equation}
where
\begin{equation*}
	\phi(u) = H(f_1 u + (1-u) f_0) - u H(f_1) - (1-u) H(f_0)
\end{equation*}
is nonnegative. The densities $f_0$ and $f_1$ correspond to the noisy channel \footnote{The function $I(A)$ denotes the indicator function throughout the paper-i.e., $I(A)=1$ if $A$ is true and zero otherwise.}:
\begin{equation*}
	\PP(Y_{n+1}=y|Z_n=z) = \left\{\begin{array}{ll} f_1(y), & z=1 \\ f_0(y), & z=0 \end{array} \right.
\end{equation*}
where $Z_n=I(X^*\in A_n) \in \{0,1\}$ is the channel input. 
The noisy channel models the conditional probability of the response to each question being correct. For the special case of a binary symmetric channel (BSC), $u^* = 1/2$ and the probabilistic bisection policy \cite{Jedynak12, CastroNowak07} becomes an optimal policy. This algorithm provides an adaptive design for the sequence of questions for all agents in the network.

In \cite{Tsiligkaridis:2013} the approach of \cite{Jedynak12} was extended to the case of multi-agent query strategies, denoted in \cite{Tsiligkaridis:2013} as collaborative 20 questions. In \cite{Tsiligkaridis:2013}, optimality conditions are derived for optimal query strategies in the collaborative multiplayer case where observations are communicated to a fusion center (or centralized controller) and were shown to generalize the probabilistic bisection policy. Two policies were studied; a sequential bisection policy for which each player responds to a single question about the state of the target, and a joint policy where all players are asked questions simultaneously. 
It was proven that the maximum entropy reduction for the sequential bisection scheme is the same as that of the jointly optimal scheme, and is given by the sum of the capacities of all the players' channels. Thus, the centralized controller is equivalent to a cascade of low-complexity controllers. 
Despite the fact that the optimal sequential policy has access to a more refined filtration, it achieves the same average performance as the optimal joint policy. This equivalence was also extended to the setting where the error channels associated with the players are unknown. 

%

\subsection{Non-Bayesian Social Learning}
In Jadbabaie, et al., \cite{Jadbabaie:2012} it is assumed that $\Theta$ denotes a finite set of possible states of the world and the objective is to study conditions for asymptotic agreement on the true state of the world, denoted by $\theta^*$. A set $\mathcal{N}=\{1,\dots,M\}$ of agents interacting over a social network (directed graph) $G=(\mathcal{N},E)$ is considered, where $E$ encodes the edges between agents. An edge connecting agent $i$ and agent $j$ is the ordered pair $(i,j)\in E$, denoting that agent $j$ has access to the belief of agent $i$. The interactions are captured by an interaction matrix $\bA$, where $a_{i,j}$ denotes the strength associated with the communication of agent $j$'s belief to agent $i$.

The belief of agent $i$ at time $t\geq 0$, defined on $\Theta$, is denoted by $p_{i,t}(\theta)$. Conditioned on the state of the world $\theta$, at each time $t\geq 1$, an observation set $\by_t=(y_{1,t},\dots,y_{M,t})$ is generated by the likelihood function $l(\cdot|\theta)$. The signal $y_{i,t}\in \mathcal{Y}$ is a private signal observed by agent $i$ at time $t$ and $\mathcal{Y}$ is a finite set. Independence across time is also assumed.

The notion of observational equivalence is key to the results derived in \cite{Jadbabaie:2012}, which are related to identifiability. Two states are observationally equivalent from the point of view of an agent if the likelihood of the two states are identical. More specifically, elements of the set $\Theta_i^\theta=\{\tilde{\theta}\in \Theta: l_i(y|\tilde{\theta})=l_i(y|\theta), \forall y \in \mathcal{Y}\}$ are observationally equivalent to state $\theta$ from the point of view of agent $i$.

The belief of each agent $i$ is updated by neighborhood averaging of the form:
\begin{equation} \label{eq:update}
	p_{i,t+1}(\theta) = a_{i,i} p_{i,t}(\theta) \frac{l_i(y_{i,t+1}|\theta)}{\mathcal{Z}_{i,t}(y_{i,t+1})} + \sum_{j \in \mathcal{N}_i} a_{i,j} p_{j,t}(\theta)
\end{equation}
where $\mathcal{N}_i=\{j \in \mathcal{N}: (j,i)\in E\}$ is the neighborhood set of agent $i$. The denominator $\mathcal{Z}_{i,t}(y_{i,t+1})$ is the normalizing factor of the Bayesian update given by $\mathcal{Z}_{i,t}(y_{i,t+1})=\sum_{\theta\in \Theta} p_{i,t}(\theta) l_i(y_{i,t+1}|\theta)$. The parameters $a_{i,i}$ are called the self-reliances that indicates the relative value of self belief, while the $a_{i,j}$ are the relative values placed on the neighbors' beliefs by agent $i$. As noted in \cite{Jadbabaie:2012}, although the first term in (\ref{eq:update}) is the Bayesian update of the local belief function, the second term is a linear combination of the neighboring beliefs. Equation (\ref{eq:update}) is not a true Bayesian belief update but is considerably simpler to implement. 

In Proposition 3 of \cite{Jadbabaie:2012}, the following assumptions are stated:
\begin{itemize}
	\item strong network connectivity (i.e., there exists a directed path from every agent to any other agent)
	\item $a_{i,i}>0, \forall i$.
	\item $\exists i$ such that $p_{i,0}(\theta^*)>0$.
	\item $\nexists \theta \neq \theta^*$ that is observationally equivalent to $\theta^*$ from the point of view of all agents in the network.
\end{itemize}
Under these assumptions, it is proven in \cite{Jadbabaie:2012} that all agents in the network learn the true state of the world almost surely-i.e., $p_{i,t}(\theta^*) \to 1$ with probability 1 for all $i\in \mathcal{N}$ as $t\to\infty$. 


\section{Decentralized Collaborative 20 Questions for Target Localization} \label{sec:algorithm}

Motivated by the work of \cite{Tsiligkaridis:2013} and \cite{Jadbabaie:2012}, we proceed as follows. Starting with a prior distribution $p_{i,0}(x)$ on the true target state $X^*$, the aim is to reach consensus on the correct state $X^*$ across the network through repeated querying and information sharing. Our proposed decentralized collaborative 20 questions target localization algorithm consists of two stages. Motivated by the optimality of the bisection rule for symmetric channels proved by Jedynak, et al., \cite{Jedynak12}, the first stage of the decentralized estimation algorithm is to bisect the posterior of each agent $i\in \mathcal{N}$ at $\hat{X}_{i,t}$ and refine its own belief through Bayes' rule. The second stage consists of each agent averaging its neighbor's beliefs and its own. This is repeated until convergence. The matrix $\bA$ contains the weights for collaboration between agents and are allowed to be zero  when there is no edge in the information sharing network $G$; if $a_{i,j}=0$, then agent $i$ cannot directly observe information from agent $j$ at any time. Algorithm 1 gives the details of the proposed two stage implementation.
\begin{algorithm}
\caption{ Decentralized Estimation Algorithm }
\label{alg:alg1}
\begin{algorithmic}[1]
\STATE \textbf{Input:}  {$G=(\mathcal{N},E), \bA=\{a_{i,j}: (i,j)\in \mathcal{N}\times \mathcal{N}\}, \{\epsilon_i: i\in \mathcal{N}\}$}
\STATE \textbf{Output:} {$\{\check{X}_{i,t}:i\in \mathcal{N}\}$}

	\STATE Initialize $p_{i,0}(\cdot)$ to be positive everywhere.
	
	\REPEAT
		\STATE For each agent $i\in \mathcal{N}$: \\
		\STATE	\quad Bisect posterior density: $\PP_{i,t}(A_{i,t})=1/2$.  \\
		\STATE	\quad Obtain (noisy) binary response $y_{i,t+1} \in \{0,1\}$. \\
	  \STATE	\quad Belief update: \\
		\begin{align} 
			p_{i,t+1}(x) &= a_{i,i} p_{i,t}(x) \frac{l_i(y_{i,t+1}|x,A_{i,t})}{\mathcal{Z}_{i,t}(y_{i,t+1})} \nonumber  \\
				&\quad + \sum_{j \in \mathcal{N}_i} a_{i,j} p_{j,t}(x), \qquad x\in \mathcal{X} \label{eq:density_update}
		\end{align}
		where the observation probability mass function (p.m.f.) is:
		\begin{align}
			l_i(y|x,A_{i,t}) &= f_1^{(i)}(y) I(x \in A_{i,t}) + f_0^{(i)}(y) I(x \notin A_{i,t}), \nonumber \\
				&\qquad y \in \mathcal{Y}  \label{eq:observation_density}
		\end{align}
		and $f_1^{(i)}(y)=(1-\epsilon_i)^{I(y=1)} \epsilon_i^{I(y=0)}, f_0^{(i)}(y)=1-f_1^{(i)}(y)$. \\
		\STATE  \quad Calculate target estimate: $\check{X}_{i,t} = \int_{\mathcal{X}} x p_{i,t}(x) dx$.
	\UNTIL {convergence}
\end{algorithmic}
\end{algorithm}

Some simplifications occur in Algorithm 1. The normalizing factor $\mathcal{Z}_{i,t}(y)$ is given by $\int_{\mathcal{X}} p_{i,t}(x) l_i(y|x,\hat{X}_{i,t}) dx$ and can be shown to be equal to $1/2$ (see proof of Lemma \ref{lemma:lemmaA} in Appendix A). The bisection query points are medians $\hat{X}_{i,t}=F_{i,t}^{-1}(1/2)$ and the observation distribution becomes:
\begin{equation*}
	l_i(y|x,\hat{X}_{i,t}) = f_1^{(i)}(y) I(x \leq \hat{X}_{i,t}) + f_0^{(i)}(y) I(x > \hat{X}_{i,t}).
\end{equation*}
where the distributions $f_z^{(i)}(\cdot)$ are defined in (\ref{eq:BSC}). We note that the conditioning on the query region $A_{i,t}$ (or query point $\hat{X}_{i,t}$) is necessary as the binary observation $\by$ is linked to the query in the 20 questions model, in which the correct answer is obtained with probability $1-\epsilon_i$ and the wrong answer is obtained with probability $\epsilon_i$ (also see Assumptions 1,2 in Section \ref{sec:convergence}). An example of this observation density is illustrated in Fig. \ref{fig:observation_density}.
\begin{figure}[ht]
	\centering
		\includegraphics[width=0.45\textwidth]{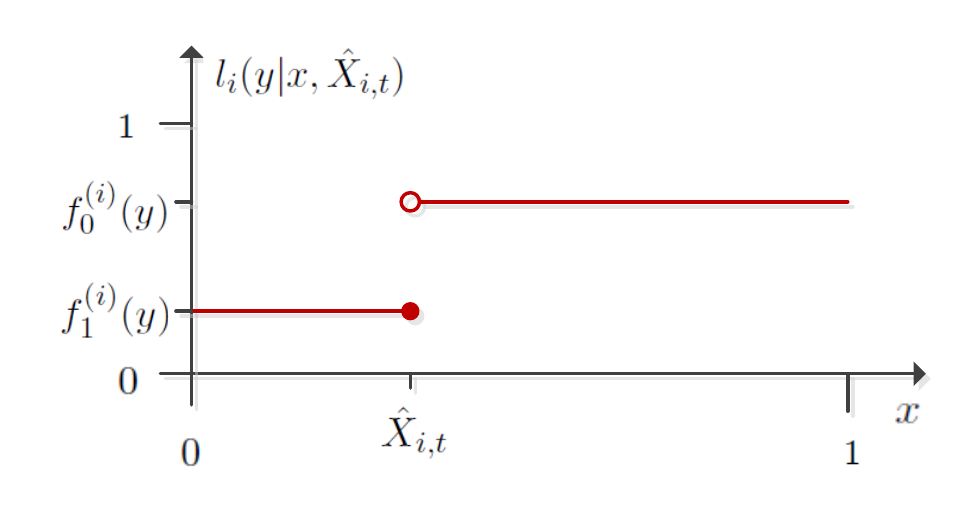}
	\caption{ An example of a time-varying observation density based on query point $\hat{X}_{i,t}$ in one dimension. }
	\label{fig:observation_density}
\end{figure}
We note two important differences between our density update (\ref{eq:density_update}) and the update (\ref{eq:update}). The density $l_i(y|x,\hat{X}_{i,t})$ depends on the query point $\hat{X}_{i,t}$, which is time-varying and as a result, the density $l_i(y|x,\hat{X}_{i,t})$ is time-varying, unlike the time-invariant case in (\ref{eq:update}). Thus, the identifiability assumptions made in \cite{Jadbabaie:2012} are not applicable for our problem. In addition the update (\ref{eq:density_update}) holds pointwise for every $x\in \mathcal{X}$ and the sequence $\{p_{i,t}(x)\}_{t\geq 0}$ may not be bounded as $t\to\infty$ for a fixed $x\in \mathcal{X}$, unlike the discrete case in (\ref{eq:update}). In the next section we consider the basic convergence properties of the decentralized estimation algorithm driven by actively controlled queries.

\section{Convergence of Decentralized Algorithm} \label{sec:convergence}

In this section convergence properties of Algorithm 1 are established under the assumptions below. The two main theoretical results,  Thm. 1 and Thm. 2, establish that the proposed algorithm attains asymptotic agreement (consensus) and asymptotic consistency, respectively. A number of technical lemmas are necessary and are proven in the appendices. A block diagram showing the interdependencies between the lemmas and theorems in this section is shown in Fig. \ref{fig:guide}.
\begin{figure}[ht]
	\centering
		\includegraphics[width=0.30\textwidth]{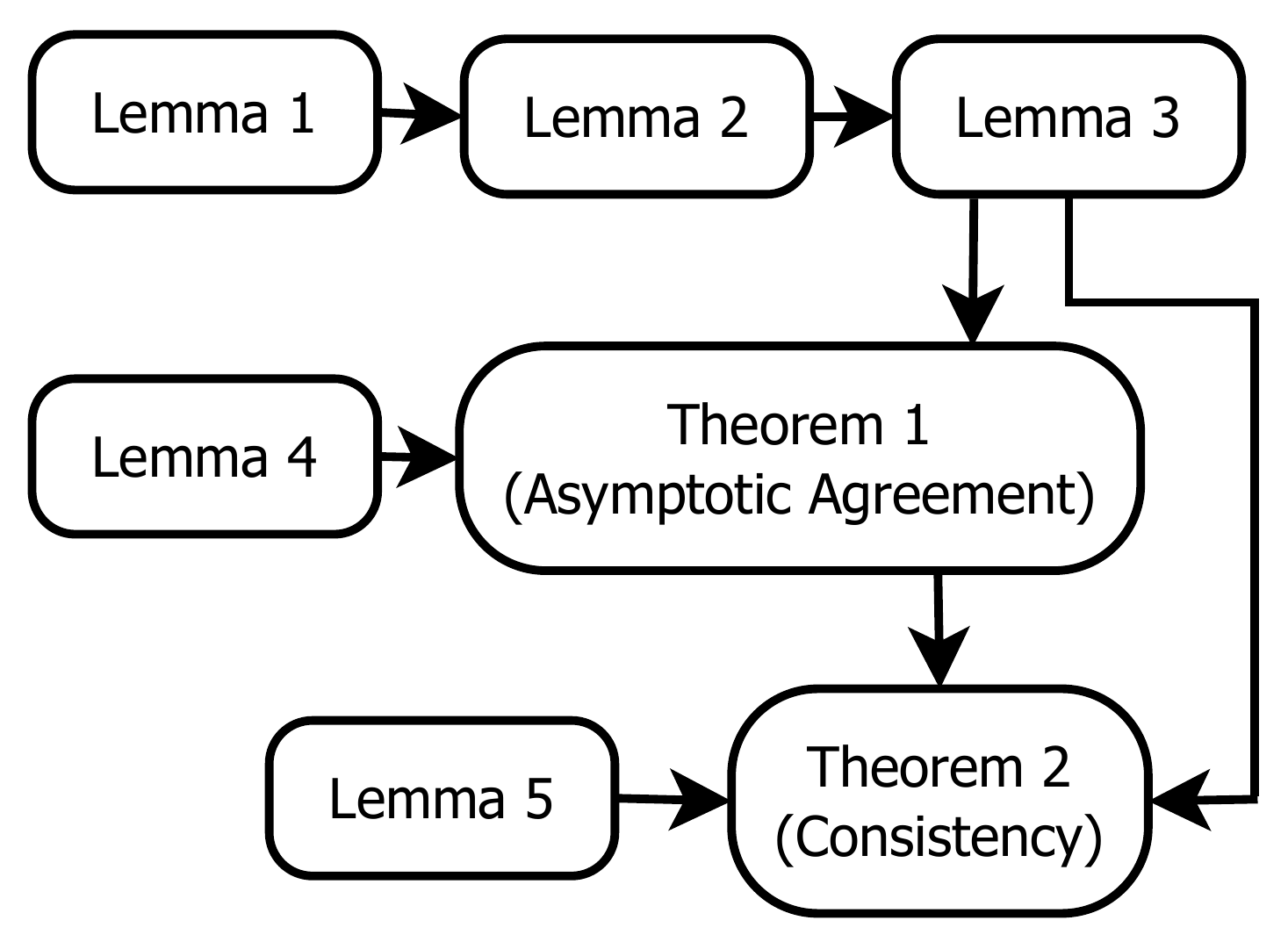}
	\caption{ The flow of the analysis for establishing convergence of the proposed decentralized 20-questions algorithm to the correct consensus limit. }
	\label{fig:guide}
\end{figure}

To simplify the analysis of Algorithm \ref{alg:alg1}, we make the following assumptions.  These assumptions are no stronger than those made in \cite{Tsiligkaridis:2013} and \cite{Jadbabaie:2012}.

\begin{assumption} (Conditional Independence) \label{assump:cond_indep}
We assume that the players' responses are conditionally independent. In particular,
\begin{equation} \label{eq:indep}
	\PP(\bY_{t+1}=\by|\mathcal{F}_t) = \prod_{i=1}^M \PP(Y_{i,t+1}=y_i|\mathcal{F}_t)
\end{equation}
and each player's response is governed by the conditional distribution:
\begin{align}
	l_i(y_i|x,A_{i,t}) &:= \PP(Y_{i,t+1}=y_i|A_{i,t},X^*=x) \nonumber \\
		&= \left\{ \begin{array}{ll} f_1^{(i)}(y_i), & x\in A_{i,t} \\ f_0^{(i)}(y_i), & x\notin A_{i,t} \end{array} \right.  \label{eq:exp}
\end{align}
\end{assumption}

\begin{assumption} (Memoryless Binary Symmetric Channels) \label{assump:BSC}
	We model the players' responses as independent (memoryless) binary symmetric channels (BSC) \cite{CoverThomas} with crossover probabilities $\epsilon_i\in (0,1/2)$. The probability mass function $f_z^{(i)}(Y_{i,t+1})=\PP(Y_{i,t+1}|Z_{i,t}=z)$ is:
\begin{equation} \label{eq:BSC}
	f_z^{(i)}(y_i) = \Bigg\{ \begin{array}{ll} 1-\epsilon_i, & y_i=z \\ \epsilon_i, & y_i\neq z \end{array}
\end{equation}
for $i=1,\dots,M, z=0,1$. The assumption $\epsilon_i<1/2$ implies that the response of each agent $i$ is almost correct.
\end{assumption}

\begin{assumption} (Strong Connectivity \& Positive Self-reliances) \label{assump:strong_connectivity}
	As in \cite{Jadbabaie:2012}, we assume that the network is strongly connected and all self-reliances $a_{i,i}$ are strictly positive. \footnote{Theorems 1 and 2 remain valid under the weaker assumption that the network is a disjoint union of strongly connected subnetworks. However, for simplicity we assume the whole network is strongly connected. } The strong connectivity assumption implies that the interaction matrix $\bA$ is irreducible. An example of a strongly connected network is shown in the figure below. 
\end{assumption}

\noindent\begin{minipage}{0.45\textwidth}
	\centering
	{\includegraphics[width=0.50\textwidth]{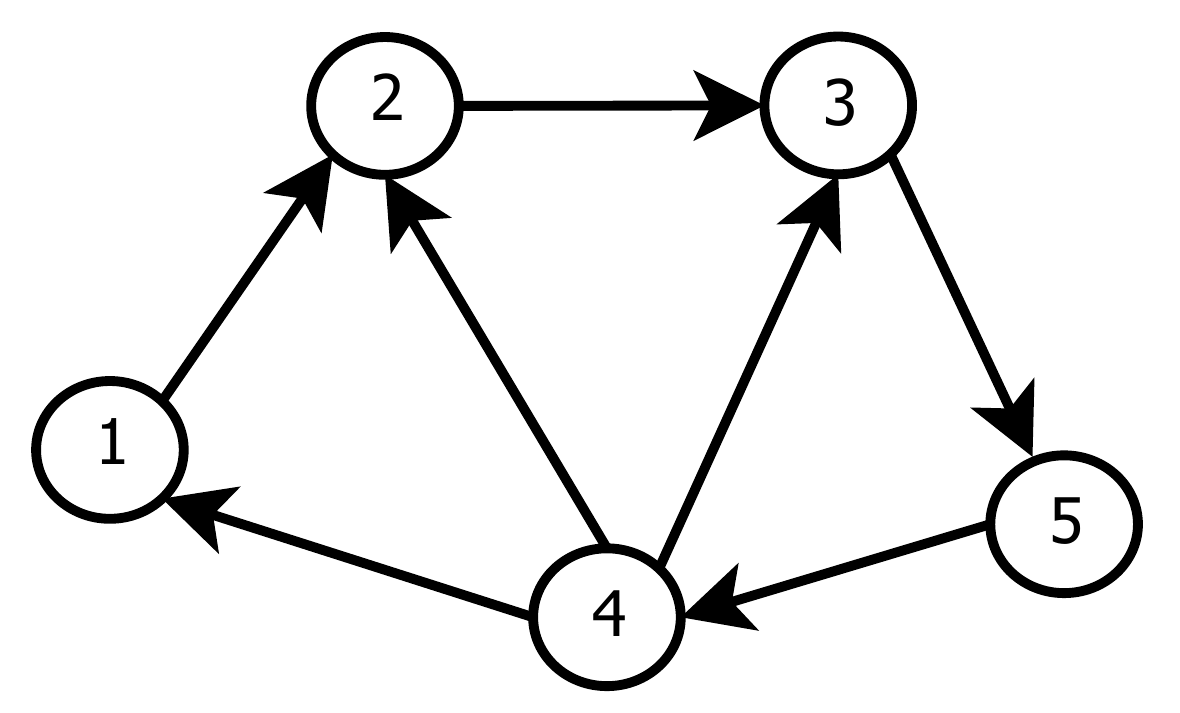}}
	\label{fig:DG}
\end{minipage}
\begin{minipage}{0.45\textwidth}
	\centering
	\begin{equation*}
			\bA = \begin{bmatrix} a_{1,1} & 0 & 0 & a_{1,4} & 0 \\ a_{2,1} & a_{2,2} & 0 & a_{2,4} & 0 \\ 0 & a_{3,2} & a_{3,3} & a_{3,4} & 0 \\ 0 & 0 & 0 & a_{4,4} & a_{4,5} \\ 0 & 0 & a_{5,3} & 0 & a_{5,5}  \end{bmatrix} 
	\end{equation*}
\end{minipage}

The starting point for studying the convergence of Algorithm 1 is Equation (\ref{eq:density_update}), which propagates the vector of belief functions forward in time. The density evolution (\ref{eq:density_update}) can be concisely written in matrix form as:
\begin{equation} \label{eq:density_update_matrix}
	\bp_{t+1}(x) = (\bA + \bD_t(x)) \bp_t(x), \qquad x\in \mathcal{X}
\end{equation}
where $\bA$ is the time-invariant interaction matrix and $\bD_t(x)$ is a diagonal time-varying matrix dependent on the responses $\by_{t+1}=(y_{1,t+1},\dots,y_{M,t+1})$, the query regions $A_{i,t} \subset \mathcal{X}$ and the state $x \in \mathcal{X}$. The $i$th diagonal entry of $\bD_t(x)$ is given by:
\begin{equation} \label{eq:D_def}
	[\bD_t(x)]_{i,i} = a_{i,i} \left( \frac{l_i(y_{i,t+1}|x,A_{i,t})}{\mathcal{Z}_{i,t}(y_{i,t+1})} - 1 \right)
\end{equation}
We remark that the convergence analysis result of Jadbabaie, et al., \cite{Jadbabaie:2012} does not apply here since the distributions $l_i(\cdot|x,A_{i,t})$ are time-varying because the query regions $A_{i,t}$ are time-varying.

Proposition \ref{prop:Amat} provides bounds on the dynamic range of $\bA\bx$, where $\bx$ is any arbitrary vector. The coefficient of ergodicity of an interaction matrix $\bA$ is defined as \cite{Seneta:1981, Ipsen:2011}:
\begin{equation} \label{eq:coeff_ergodicity}
	\tau_1(\bA) = \frac{1}{2} \max_{i\neq j} \nn \bA^T(\be_i-\be_j)\nn_1 = \frac{1}{2} \max_{i\neq j} \sum_{l=1}^M |a_{i,l}-a_{j,l}|
\end{equation}
This coefficient does not exceed $1$. The most non-ergodic interaction matrix is the identity matrix $\bA=\bI_M$, for which $\tau_1(\bA)=1$ and there is no information sharing. As another example, a matrix $\bA$ with fixed self-reliances $\alpha \in (0,1)$ and uniform off-diagonal weights-i.e., $\frac{1-\alpha}{M-1}$, has $\tau_1(\bA)=|\alpha-\frac{1-\alpha}{M-1}|$.

\begin{proposition} (Contraction Property of $\bA$) \label{prop:Amat}
	Assume $\bA=\{a_{i,j}\}$ is a $M\times M$ stochastic matrix. Let $\bx$ be an arbitrary non-negative vector. Then, we have for all pairs $(i,j)$:
	\begin{equation*}
		[\bA\bx]_i-[\bA\bx]_j \leq \tau_1(\bA) \left( \max_i x_i - \min_i x_i \right)
	\end{equation*}
\end{proposition}
For a proof, see Theorem 3.1 in \cite{Seneta:1981}.

	%


Note that irreducibility of the matrix $\bA$ implies that there exists $r$ such that $\bA^r$ is a stochastic matrix with positive entries \cite{Seneta:1981}, further implying that $\tau_1(\bA^r)<1$. Therefore, under Assumption \ref{assump:strong_connectivity}, Proposition \ref{prop:Amat} will establish a contraction property required for the convergence proof of Theorem 1.

Next, we recall a tight smooth approximation to the non-smooth maximum and minima operators. Similar results have appeared in Prop. 1 in \cite{Chen:2013} and p. 72 in \cite{Boyd:ConvexOptimization}. Consensus will be proven in Theorem 1 by showing that the dynamic range of the beliefs of the network converges to zero asymptotically as the number of algorithm iterations grow. Since the dynamic range involves taking the difference between the maximum and minimum of the CDF's evaluated at an arbitrary point, Proposition \ref{prop:lse} will be used in the proof of Theorem 1 to approximate the maximum and minimum involved.
\begin{proposition} \label{prop:lse} (Tight Smooth Approximation to Maximum/Minimum Operator)
	Let $\ba \in \RR^M$ be an arbitrary vector. Then, we have for all $k>0$:
	\begin{equation} \label{eq:lse_approx_max}
		\max_i a_i \leq \frac{1}{k} \log\left( \sum_{i=1}^M \exp(k a_i) \right) \leq \max_i a_i + \frac{\log M}{k}
	\end{equation}
	and
	\begin{equation} \label{eq:lse_approx_min}
		\min_i a_i \geq -\frac{1}{k} \log\left( \sum_{i=1}^M \exp(-k a_i) \right) \geq \min_i a_i - \frac{\log M}{k}
	\end{equation}
\end{proposition}

The next lemma shows that the integrated innovation term in (\ref{eq:density_update_matrix}) has a conditional mean of zero, which will be used to prove the martingale property in Lemma \ref{lemma:lemmaB}.
\begin{lemma} \label{lemma:lemmaA}
	Consider Algorithm \ref{alg:alg1}. Let $B \in \mathcal{B}(\mathcal{X})$. Then, we have:
	\begin{equation*}
		\EE\left[\int_B \bD_t(x) \bp_t(x) dx \Bigg| \mathcal{F}_t \right] = 0.
	\end{equation*}
	where $\bD_t(x)$ was defined in (\ref{eq:D_def}).
\end{lemma}
\begin{proof}
	See Appendix A.
\end{proof}

Lemma \ref{lemma:lemmaB} shows that a positive linear combination of the integrated posterior distributions forms a martingale sequence. This will allow us to use the martingale convergence theorem to prove Lemma \ref{lemma:lemmaC} (a key lemma in proving the main convergence Theorems 1 and 2).
\begin{lemma} \label{lemma:lemmaB}
	Consider Algorithm \ref{alg:alg1}. Let $B \in \mathcal{B}(\mathcal{X})$. Then, we have $\EE[\bv^T\bP_{t+1}(B)|\mathcal{F}_t]=\bv^T\bP_t(B)$ for some positive vector $v\succ 0$, and $\lim_{t\to\infty} \bv^T\bP_{t}(B)$ exists almost surely.
\end{lemma}
\begin{proof}
	See Appendix B.
\end{proof}

The next lemma obtains an asymptotic convergence result on the CDF of each agent in the network, which will be crucial for proving the main Theorems 1 and 2.
\begin{lemma} \label{lemma:lemmaC}
	Consider Algorithm \ref{alg:alg1}. Let $b\in [0,1]$. Then, we have: 
		\begin{equation} \label{eq:mu_convergence}
			\mu_{i,t}(b) := \min\left\{ F_{i,t}(b), 1-F_{i,t}(b) \right\} \stackrel{a.s.}{\longrightarrow} 0
		\end{equation}
		as $t\to\infty$.
\end{lemma}
\begin{proof}
	See Appendix C.
\end{proof}

Define the dynamic range (with respect to all agents in the network) of the posterior probability that $X^*$ lies in set $B\subset \mathcal{X}$:
\begin{equation} \label{eq:dyn_range}
	V_t(B) = \max_i \PP_{i,t}(B) - \min_i \PP_{i,t}(B)
\end{equation}
Also, define the innovation:
\begin{equation*}
	d_{i,t+1}(B) = \left[ \int_B \bD_t(x)\bp_t(x) dx \right]_i = \int_B [\bD_t(x)]_{i,i} p_{i,t}(x) dx
\end{equation*}

We next prove a lemma that shows that the dynamic range $V_t(B)$ has a useful upper bound. 
\begin{lemma} \label{lemma:lemmaD}
	Consider Algorithm \ref{alg:alg1}. Let $B=[0,b]$ with $b \leq 1$. Then, for all $r\in \NN$:
	\begin{align} 
		V_{t+r}(B) &\leq \tau_1(\bA^r) V_t(B) \nonumber \\
		&\quad + \sum_{k=0}^{r-1} \left( \max_i d_{i,t+r-k}(B) - \min_i d_{i,t+r-k}(B) \right) \label{eq:V_bound}
	\end{align}
	In addition, there exists a finite $r\in \NN$ such that $\tau_1(A^r)<1$.	
\end{lemma}
\begin{proof}
	See Appendix D.
\end{proof}

To show convergence of the integrated beliefs of all agents in the network to a common limiting belief, it suffices to show $V_t(B) \stackrel{i.p.}{\to} 0$. While this method of proof does not allow identification of the limiting belief, it shows a global equilibrium exists and yields insight into the rate of convergence through the ergodicity properties of $\bA$. The structure of the limiting belief is given in Theorem \ref{thm:thmB}. Theorem \ref{thm:thmA} shows convergence of asymptotic beliefs to a common limit.
\begin{theorem} \label{thm:thmA} 
	Consider Algorithm \ref{alg:alg1} and let the assumptions in Sec. V.A hold. Let $B=[0,b]$, $b\leq 1$. Then, consensus of the agents' beliefs is asymptotically achieved across the network:
	\begin{equation*}
		V_t(B) = \max_{i} \PP_{i,t}(B) - \min_i \PP_{i,t}(B) \stackrel{i.p.}{\longrightarrow} 0
	\end{equation*}
	as $t\to\infty$.
\end{theorem}
\begin{proof}
	See Appendix E.
\end{proof}

Theorem \ref{thm:thmA} establishes that Algorithm 1 produces belief functions that become identical over all agents. This establishes asymptotic consensus among the beliefs, i.e., that as time goes on all agents come to agreement about the uncertainty in the target state. It remains to show the limiting belief is in fact concentrated at the true target state $X^*$, as stated in Thm. \ref{thm:thmB}.
\begin{lemma} \label{lemma:lemmaE}
	Consider Algorithm \ref{alg:alg1}. Let $\bv$ be the left eigenvector of $\bA$ corresponding to the unit eigenvalue. Assume that for all agents $i$, $p_{i,0}(X^*)>0$. Then, the posteriors evaluated at the true target state $X^*$ have the following asymptotic behavior:
	\begin{equation*}
		\liminf_{t\to\infty} \frac{1}{t} \sum_{i=1}^M v_i \log(p_{i,t}(X^*)) \geq \sum_{i=1}^M v_i a_{i,i} C(\epsilon_i) = K(\bepsilon) \quad (a.s.)
	\end{equation*}
	where $C(\epsilon)$ is the capacity of the BSC.
\end{lemma}
\begin{proof}
	See Appendix F.
\end{proof}

Now, we are ready to prove the main consistency result of the asymptotic beliefs. The proof is based on the consensus result of Theorem \ref{thm:thmA}.
\begin{theorem} \label{thm:thmB} 
	Consider Algorithm \ref{alg:alg1} and let the assumptions in Sec. V.A hold. Let $b\in [0,1]$. Then, we have for each $i\in \mathcal{N}$:
		\begin{equation*}
			F_{i,t}(b) \stackrel{i.p.}{\longrightarrow} F_{\infty}(b) =\left\{ \begin{array}{ll} 0, & b<X^* \\ 1, & b>X^* \end{array} \right.
		\end{equation*}
		as $t\to\infty$. In addition, for all $i\in \mathcal{N}$:
		\begin{equation*}
			\check{X}_{i,t}:=\int_{x=0}^1 x p_{i,t}(x) dx \stackrel{i.p.}{\longrightarrow} X^*
		\end{equation*}
\end{theorem}
\begin{proof}
	See Appendix G.
\end{proof}

Finally, we have a corollary that generalizes Theorem \ref{thm:thmB} to sets $B$ that are finite unions of intervals.
\begin{corollary} \label{cor:corB}
	Consider Algorithm \ref{alg:alg1}. Let $B=\cup_{k=1}^K I_k \in \mathcal{B}([0,1])$ be a finite union of disjoint intervals $I_k=[a_k,b_k)$, where $0\leq a_k < b_k \leq 1$. Then, for each $i\in \mathcal{N}$:
	\begin{equation*}
			\PP_{i,t}(B) \stackrel{i.p.}{\longrightarrow} \left\{ \begin{array}{ll} 0, & X^*\notin B \\ 1, & X^*\in B \end{array} \right.
		\end{equation*}
		as $t\to\infty$.
\end{corollary}
\begin{proof}
	The proof follows by noting that $\PP_{i,t}(B) = \PP_{i,t}(\cup_k I_k) = \sum_k \PP_{i,t}(I_k)=\sum_k F_{i,t}(b_k)-F_{i,t}(a_k)$ and using Theorem \ref{thm:thmB}.
\end{proof}

\section{Experimental Validation} \label{sec:simulations}

This section presents simulations that validate the theory in Section \ref{sec:convergence} and demonstrate the benefits of the proposed decentralized 20 questions method. As expected from Theorems 1 and 2, the estimation error converges to zero for all agents in the network as the number of algorithm iterations grows, implying  correct learning of the true target state across the network.

The instantaneous squared residual error for agent $i$ was calculated using $\text{SE}_{i,t}=(\hat{X}_{i,t}-X^*)^2$ for the $t$th Monte Carlo trial. There are a total of $T$ experimental trials. The min, max and average RMSE metrics were calculated as:
\begin{align}
	\text{RMSE}_{\text{min}} &= \sqrt{ \frac{1}{T} \sum_{t=1}^T \min_i \text{SE}_{i,t} }                   \label{eq:minRMSE} \\
	\text{RMSE}_{\text{max}} &= \sqrt{ \frac{1}{T} \sum_{t=1}^T \max_i \text{SE}_{i,t} }                   \label{eq:maxRMSE} \\
	\text{RMSE}_{\text{avg}} &= \sqrt{ \frac{1}{T} \sum_{t=1}^T \frac{1}{M} \sum_{i=1}^M \text{SE}_{i,t} } \label{eq:avgRMSE}
\end{align}
The min and max metrics, $\text{RMSE}_{\text{min}}$ and $\text{RMSE}_{\text{max}}$, represent the worst and best performance over all the agents and the average over all agents is denoted $\text{RMSE}_{\text{avg}}$.

For comparison we implement the centralized fully Bayesian estimation algorithm, which requires full knowledge of the error probabilities of all agents. We make use of the basic equivalence principle derived in \cite{Tsiligkaridis:2013}, and implement the centralized method using a series of bisections (one per agent). The equivalence principle shows that this sequential bisection algorithm achieves the same performance as the jointly optimal algorithm on average.

We consider the performance of Algorithm \ref{alg:alg1} for random graphs based on the Erd\"{o}s-R\'{e}nyi construction. Let $p\in (0,1)$ denote the probability of two nodes being joined with an edge, and $M$ denote the number of nodes in the graph. The set $\mathcal{G}(M,p)$ denotes the class of all undirected random graphs based on the Erd\"{o}s-R\'{e}nyi construction, and $\mathcal{G}^{IR}(M,p) \subset \mathcal{G}(M,p)$ denotes the subclass of all irreducible graphs. We consider the ensemble-average RMSE's, $\EE[\text{RMSE}_{\text{min}}],\EE[\text{RMSE}_{\text{avg}}],\EE[\text{RMSE}_{\text{max}}]$, and approximate them by random sampling from $\mathcal{G}^{IR}(M,p)$. An example graph $G \in \mathcal{G}^{IR}(100,0.05)$ is shown in Fig. \ref{fig:graph_ER_p0_05}.
\begin{figure}[ht]
	\centering
		\includegraphics[width=0.25\textwidth]{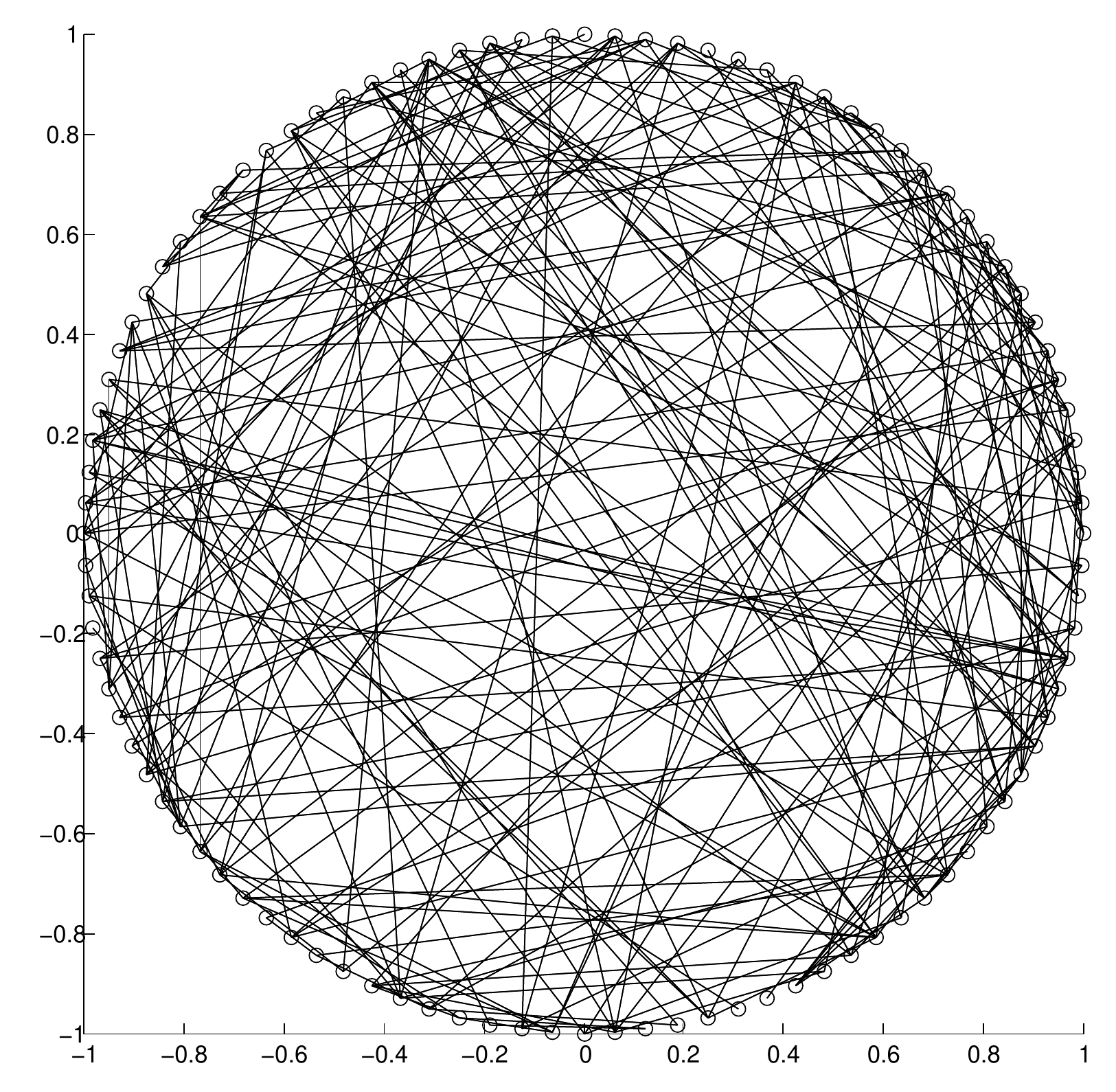}
	\caption{A realization, $G$, from the ensemble of Erd\"{o}s-R\'{e}nyi irreducible graphs $\mathcal{G}^{IR}(100,0.05)$.}
	\label{fig:graph_ER_p0_05}
\end{figure}
Figures \ref{fig:average_RMSE_p0_05}, \ref{fig:max_RMSE_p0_05} and \ref{fig:min_RMSE_p0_05} show the average, worst-case and best-case RMSE for a graph of $M=100$ agents with a connectivity probability of $p=0.05$. For this experiment the number of reliable agents $M_{\text{rel}}$, i.e., with error probability $\epsilon_i=0.05$, is either $0$ or $1$, and the unreliable agents, i.e., with error probability $\epsilon_i=0.45$, are the $M-M_{\text{rel}}$ remaining ones. The self-reliance parameter of the reliable agent was set to $0.95$ and the self-reliance of the unreliable agents was set to $0.6$. The rest of the parameters were made equal such that each row of $\bA$ sums to unity. The error performance is averaged over $500$ Monte Carlo runs.


In terms of average and worst-case RMSE performance of the network, Figures \ref{fig:average_RMSE_p0_05} and \ref{fig:max_RMSE_p0_05} show that the decentralized estimation algorithm with information sharing uniformly outperforms the algorithm without information sharing over all iterations. Of course, as Fig. \ref{fig:min_RMSE_p0_05} shows, this naturally occurs at a penalty for the best-case RMSE performance, which is attained by discounting all of the agents except for the single reliable one. 
The phenomenon of local information aggregation leads to sophisticated global behavior; convergence of all agents' estimates towards the correct target.

Furthermore, we observe from Figures \ref{fig:average_RMSE_p0_05}-\ref{fig:min_RMSE_p0_05} that if reliable agents are introduced into the network, the average RMSE and worst-case RMSE are both dramatically reduced and the best-case RMSE penalty is significantly diminished. Even one reliable agent injected in a large set of unreliable agents greatly enhances the network-wide estimation performance if information sharing is implemented. This is entirely due to the decentralized nature of the algorithm; the good information of the reliable agent is spread around the network, affecting in a positive manner all the agents in the network.
\begin{figure}[ht]
	\centering
		\includegraphics[width=0.50\textwidth]{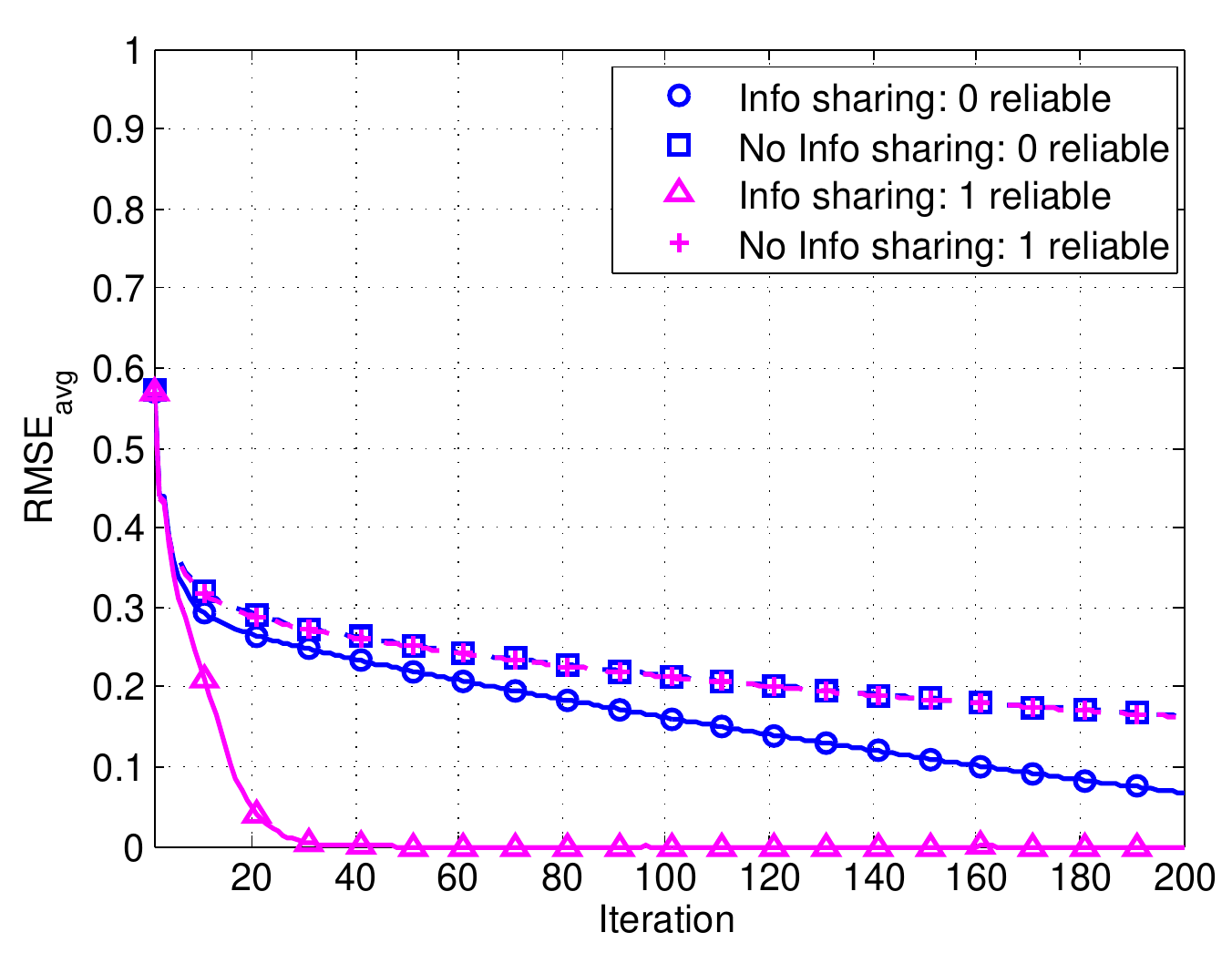}
	\caption{ Average RMSE as a function of algorithm iteration. With information sharing between a single reliable agent ($\epsilon=0.05$) and the unreliable agents ($\epsilon=0.45$) the proposed decentralized 20 questions algorithm attains a much lower RMSE (solid magenta line) than without information sharing.  The RMSE was computed as the average network RMSE (\ref{eq:avgRMSE}), and is further averaged over 10 realizations of irreducible Erd\"{o}s-R\'{e}nyi graphs drawn from $\mathcal{G}^{IR}(100,0.05)$. }
		\label{fig:average_RMSE_p0_05}
\hfill
	\centering
	\center{
		\includegraphics[width=0.50\textwidth]{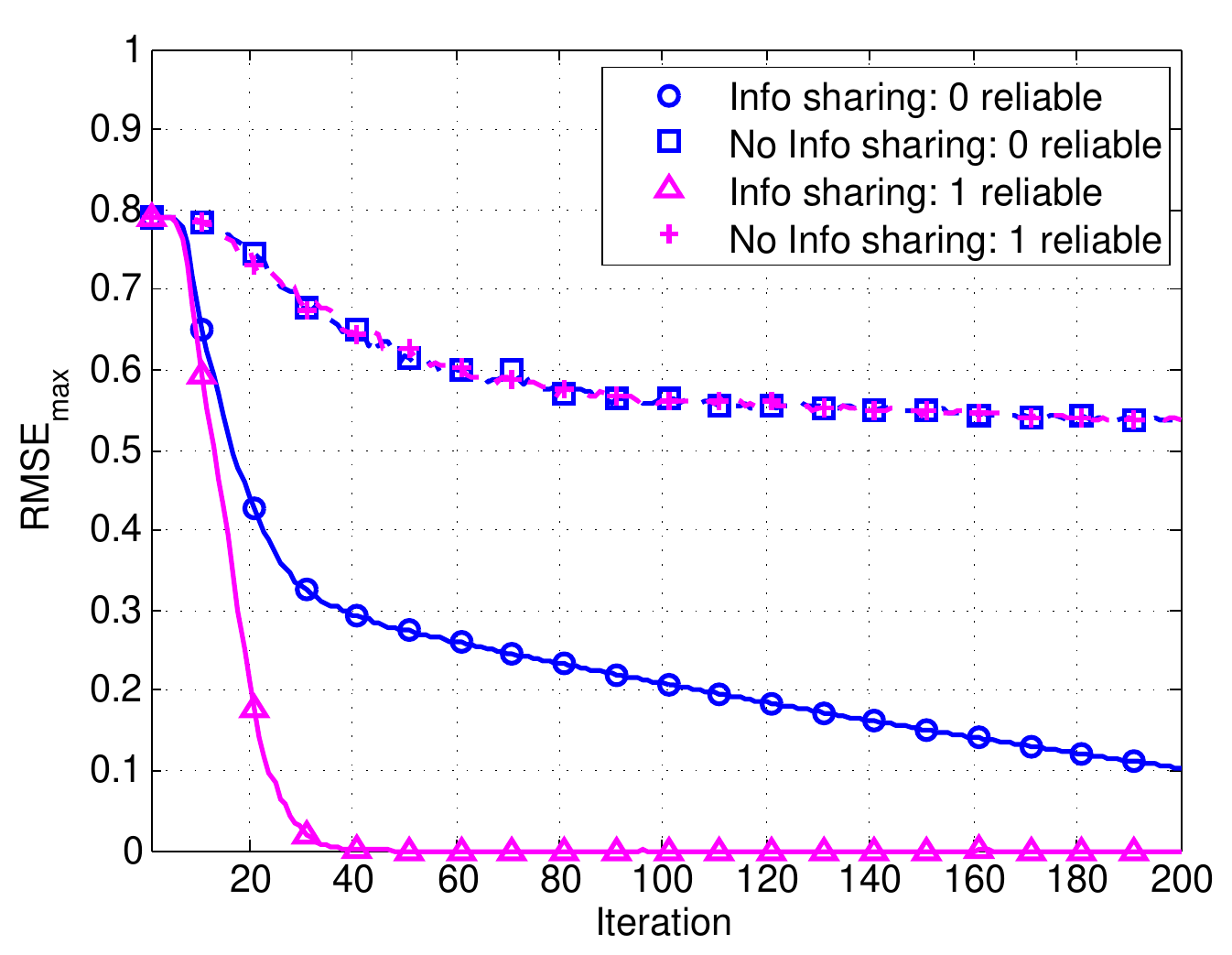}
	\caption{ Same as in Fig. \ref{fig:average_RMSE_p0_05} except we show the worst-case RMSE, computed using (\ref{eq:maxRMSE}). }
		\label{fig:max_RMSE_p0_05}
		}
\hfill
	\centering
	\center{
		\includegraphics[width=0.50\textwidth]{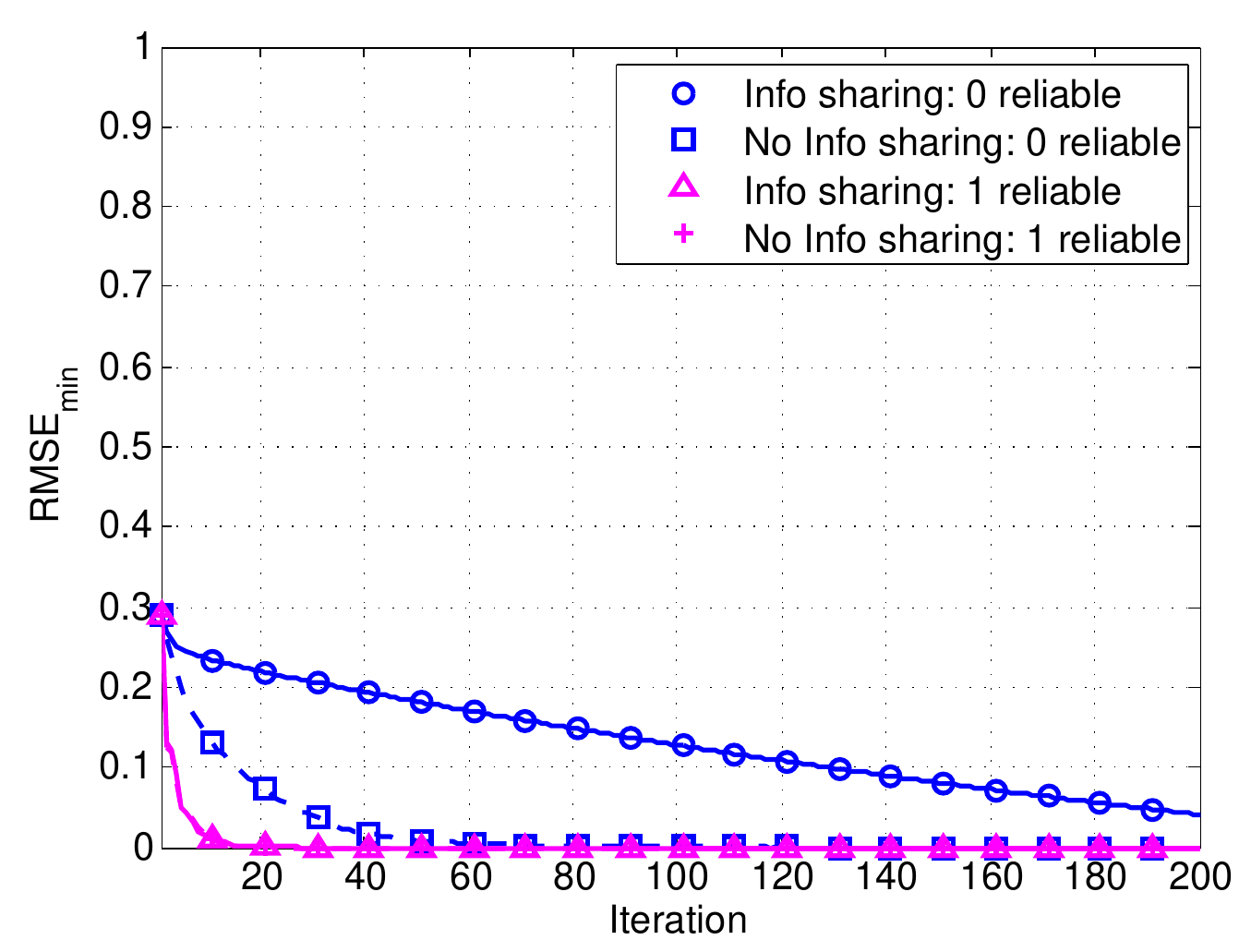}
	\caption{ Same as in Fig. \ref{fig:average_RMSE_p0_05} except we show the best-case RMSE, computed using (\ref{eq:minRMSE}). }
		\label{fig:min_RMSE_p0_05}
		}
\end{figure}

Figures \ref{fig:RMSE_p0_05_0good} and \ref{fig:RMSE_p0_05_1good} compare the decentralized algorithm's performance to the centralized one, for the cases of $M_{\text{rel}}=0$ (no reliable agents) and $M_{\text{rel}}=1$ (a single reliable agent), respectively. We note that the best-case RMSE performance of the decentralized algorithm when $M_{\text{rel}}=1$ is comparable with the performance of the centralized algorithm. On the other hand, with no reliable agents ($M_{\text{rel}}=0$), the centralized Bayesian solution provides a significant performance advantage. Under the conditions of this experiment, the presence of a single reliable agent yields a significant performance gain.
\begin{figure}[ht]
	\centering
		\includegraphics[width=0.50\textwidth]{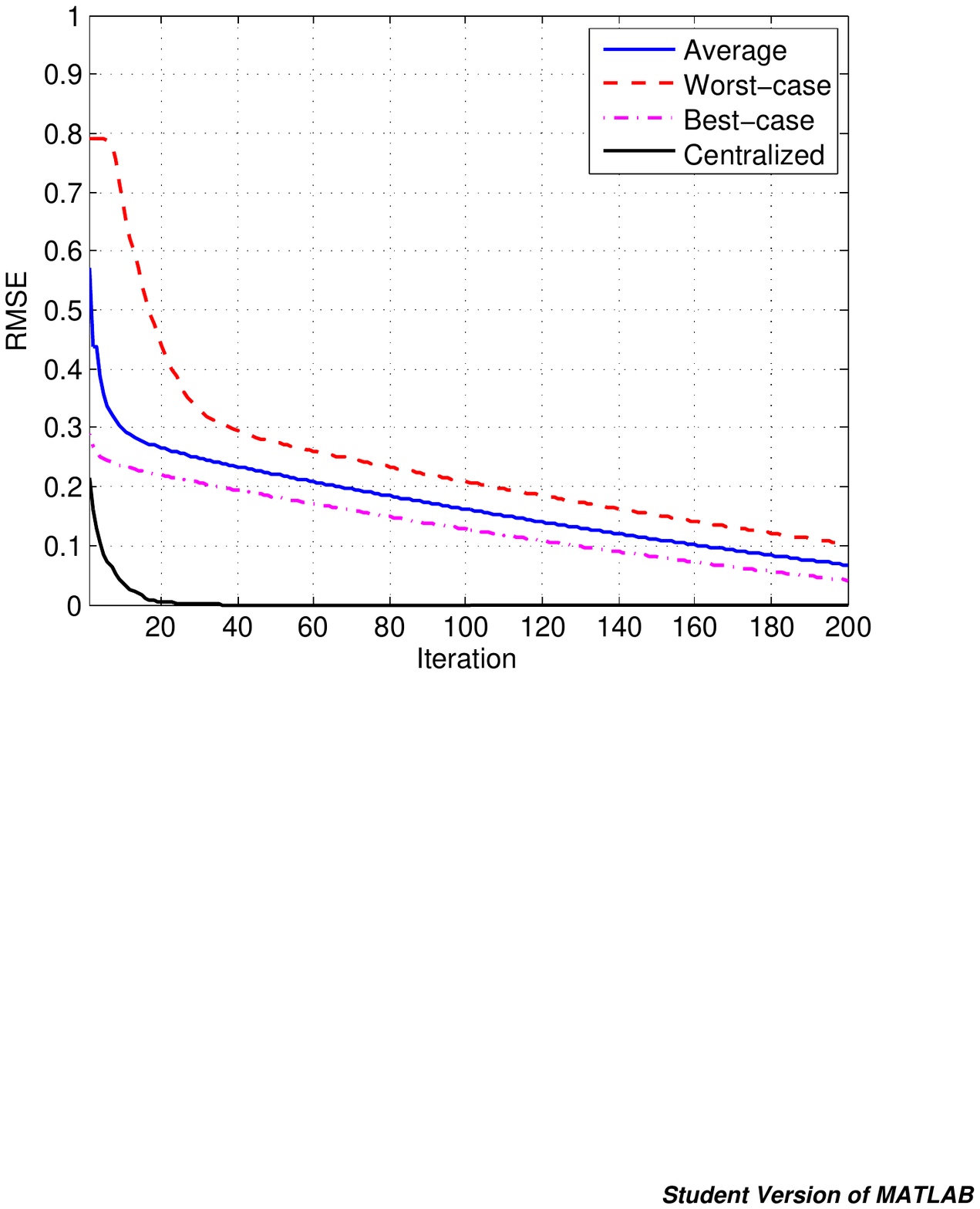}
	\caption{ Average RMSE performance as a function of algorithm iteration. Here, all agents are unreliable; i.e., have an error probability $\epsilon_i=0.45$. The centralized 20 questions strategy (bottom curve) is significantly better than the proposed decentralized strategy with information sharing, in terms of best, average and worst case RMSE, since there are no reliable agents. The average, worst-case and best-case RMSE was computed using (\ref{eq:minRMSE})-(\ref{eq:avgRMSE}), and is further averaged over 10 realizations of irreducible Erd\"{o}s-R\'{e}nyi graphs drawn from $\mathcal{G}^{IR}(100,0.05)$. }
		\label{fig:RMSE_p0_05_0good}
\hfill
	\centering
	\center{
		\includegraphics[width=0.50\textwidth]{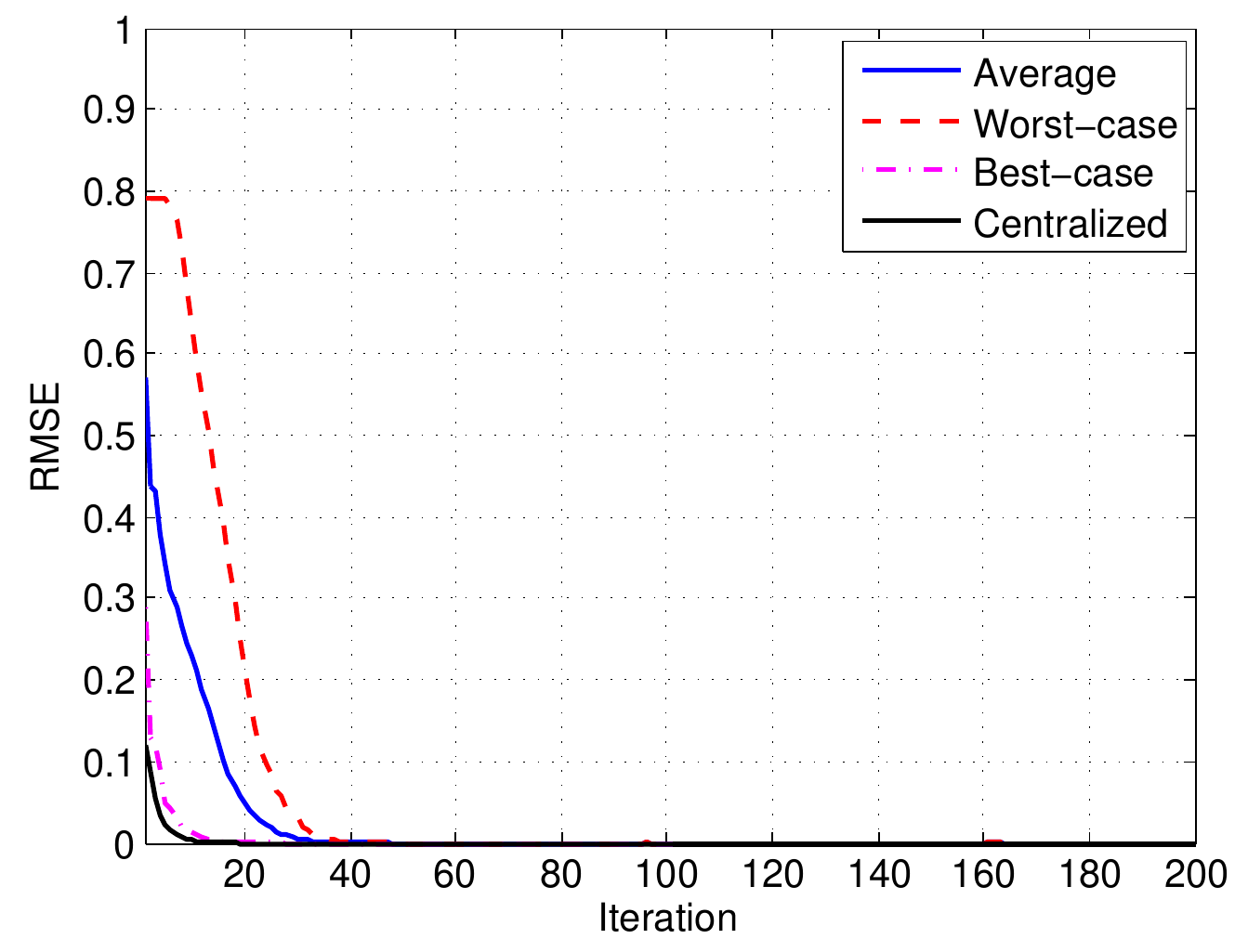}
	\caption{ Same as Fig. \ref{fig:RMSE_p0_05_0good} except that there is a single reliable agent with error probability $\epsilon_1=0.05$ and the rest of the agents are all unreliable with error probability $\epsilon_i=0.45$. As compared to Fig. \ref{fig:RMSE_p0_05_0good}, the presence of even one reliable agent makes the proposed decentralized information sharing algorithm have performance close to that of the centralized algorithm. }
		\label{fig:RMSE_p0_05_1good}
		}
\end{figure}

\section{Conclusion} \label{sec:conclusions}
We proposed a solution to the problem of decentralized 20 questions with noise and illustrated several benefits of information sharing as compared to no information sharing. At each iteration of our proposed decentralized information sharing algorithm, agents query and respond based on their local beliefs and average information through their neighbors. 
Asymptotic convergence properties of the agents' posterior distributions were derived, showing that they reach consensus to the true state. Numerical experiments were presented to validate the convergence properties of the algorithm.

We note that lemmas \ref{lemma:lemmaA}, \ref{lemma:lemmaB} and \ref{lemma:lemmaD} hold for any dimension $d\geq 1$, while the rest of the analysis in the paper holds for a scalar target state. Generalization of Lemma \ref{lemma:lemmaC} to state dimension $d>1$ is an open problem, see the discussion in Waeber, et al., \cite{Waeber:2013} in the context of extending the convergence theory of the probabilistic bisection algorithm (PBA) to higher dimensions. Several other interesting open problems  arise from this decentralized 20 questions algorithm including analysis of the rate of convergence and tuning of the interaction matrix weights in some optimal manner to improve estimation performance.

\newpage

\appendices

\section{Proof of Lemma \ref{lemma:lemmaA}}
\begin{proof}
	Without loss of generality, fix $i \in \mathcal{N}$. From direct substitution and integration, we have:
	\begin{align*}
		\int_B &[\bD_t(x)]_{i,i} p_{i,t}(x) dx \\
			&= a_{i,i} \left( \frac{\int_B l_i(y_{i,t+1}|x,A_{i,t}) p_{i,t}(x) dx}{\mathcal{Z}_{i,t}(y_{i,t+1})} - \int_B p_{i,t}(x) dx \right) \\
			&= a_{i,i} \left(2 \int_B l_i(y_{i,t+1}|x,A_{i,t}) p_{i,t}(x) dx - \PP_{i,t}(B) \right) \\
	\end{align*}
	where we used the fact that $\mathcal{Z}_{i,t}(y)=1/2$ for all $y\in \mathcal{Y}$. This follows from the probabilistic bisection property:
	\begin{align*}
		&\mathcal{Z}_{i,t}(y) \\
			&= \int_{\mathcal{X}} p_{i,t}(x) \left( f_1^{(i)}(y) I(x\in A_{i,t}) + f_0^{(i)}(y) I(x\notin A_{i,t}) \right) dx \\
			&= f_1^{(i)}(y) \PP_{i,t}(A_{i,t}) + f_0^{(i)}(y) \PP_{i,t}(A_{i,t}^c) \\
			&= f_1^{(i)}(y) (1/2) + f_0^{(i)}(y) (1/2) \\
			&= 1/2
	\end{align*}
	where we used the fact $f_1^{(i)}(y) + f_0^{(i)}(y) = 1$. From the definition of $l_i(y|x,A_{i,t})$, it follows that:
	\begin{align*}
		&\int_B [\bD_t(x)]_{i,i} p_{i,t}(x) dx = a_{i,i} \Bigg(2 \Big(f_1^{(i)}(y_{i,t+1}) \PP_{i,t}(B\cap A_{i,t}) \\
			&\quad + f_0^{(i)}(y_{i,t+1}) \PP_{i,t}(B \cap A_{i,t}^c) \Big) - \PP_{i,t}(B) \Bigg)
	\end{align*}
	Taking the conditional expectation of both sides, we obtain:
	\begin{align*}
		&\EE\left[\int_B [\bD_t(x)]_{i,i} p_{i,t}(x) dx\Bigg|\mathcal{F}_t\right] \\
			&= a_{i,i} \Bigg(2 \EE\Big[ f_1^{(i)}(Y_{i,t+1}) \PP_{i,t}(B\cap A_{i,t}) \\
			&\quad + f_0^{(i)}(Y_{i,t+1}) \PP_{i,t}(B \cap A_{i,t}^c) |\mathcal{F}_t\Big] - \PP_{i,t}(B) \Bigg) \\
			&= a_{i,i} \Bigg(2 \sum_{y=0}^1 \left( f_1^{(i)}(y) \PP_{i,t}(B\cap A_{i,t}) + f_0^{(i)}(y) \PP_{i,t}(B \cap A_{i,t}^c) \right) \\
			&\quad \times \PP(Y_{i,t+1}=y|\mathcal{F}_t) - \PP_{i,t}(B) \Bigg) \\
			&= a_{i,i} \Bigg(\sum_{y=0}^1 \left( f_1^{(i)}(y) \PP_{i,t}(B\cap A_{i,t}) + f_0^{(i)}(y) \PP_{i,t}(B \cap A_{i,t}^c) \right) \\
			&\qquad - \PP_{i,t}(B) \Bigg) \\
			&= a_{i,i} \left( \left( \PP_{i,t}(B\cap A_{t,i}) + \PP_{i,t}(B \cap A_{i,t}^c) \right) - \PP_{i,t}(B) \right) \\
			&= a_{i,i} \left( \PP_{i,t}(B) - \PP_{i,t}(B) \right) = 0
	\end{align*}
	where we used the fact that under the probabilistic bisection, $\PP(Y_{i,t+1}=y|\mathcal{F}_t)=1/2$ for all $y$. This follows from:
	\begin{align*}
		\PP(&Y_{i,t+1}=y|\mathcal{F}_t) \\
			&= \sum_{z=0}^1 \PP(Y_{i,t+1}=y|Z_{i,t}=z,\mathcal{F}_t) \PP(Z_{i,t}=z|\mathcal{F}_t) \\
			&= \PP(Y_{i,t+1}=y|Z_{i,t}=0) \PP(Z_{i,t}=0|\mathcal{F}_t) \\
			&\quad + \PP(Y_{i,t+1}=y|Z_{i,t}=1) \PP(Z_{i,t}=1|\mathcal{F}_t)  \\
			&= f_0(y) \PP(X^* \notin A_{i,t} |\mathcal{F}_t) + f_1(y) \PP(X^* \in A_{i,t}|\mathcal{F}_t) \\
			&= f_0(y) \PP_{i,t}(A_{i,t}^c) + f_1(y) \PP_{i,t}(A_{i,t}) \\
			&= 1/2
	\end{align*}
	Since $i$ was arbitrarily chosen, the proof is complete.
\end{proof}

\section{Proof of Lemma \ref{lemma:lemmaB}}
\begin{proof}
	From strong connectivity (i.e., Assumption 3), it follows that $\bA$ is an irreducible stochastic matrix. Thus, there exists a left eigenvector $\bv \in \RR^M$ with strictly positive entries corresponding to a unit eigenvalue-i.e., $\bv^T=\bv^T\bA$ \cite{Berman:1979}.

	Integrating (\ref{eq:density_update_matrix}) and left-multiplying by $\bv^T$:
	\begin{align}
		&\quad \bv^T\int_B \bp_{t+1}(x)dx = \bv^T\bA \int_B \bp_t(x) dx + \bv^T \int_B \bD_t(x)\bp_t(x) dx \nonumber \\
		&\Leftrightarrow \bv^T \bP_{t+1}(B) = \bv^T \bP_t(B) + \sum_{i=1}^M v_i \int_B [\bD_t(x)]_{i,i} p_{i,t}(x) dx \label{eq:eq1}
	\end{align}
	Taking the conditional expectation of both sides and using Lemma \ref{lemma:lemmaA}, we obtain $\EE[\bv^T\bP_{t+1}(B)|\mathcal{F}_t] = \bv^T\bP_t(B)$. Thus, the process $\{\bv^T\bP_t(B):t\geq 0\}$ is a martingale with respect to the filtration $\mathcal{F}_t$. We note that it is bounded below by zero and above by $\nn \bv \nn_1$ almost surely. From the martingale convergence theorem \cite{Billingsley:2012}, it follows that it converges almost surely.	
\end{proof}

\section{Proof of Lemma \ref{lemma:lemmaC}}
\begin{proof}
	Define the tilted measure variable 
	\begin{equation*}
		\zeta_{t}(B) = \exp(\bv^T\bP_t(B)).
	\end{equation*}
	From Lemma \ref{lemma:lemmaB} and Jensen's inequality, it follows that
	\begin{equation*}
		\EE[\zeta_{t+1}(B)|\mathcal{F}_t] \geq \zeta_t(B)
	\end{equation*}
	so the process $\{\zeta_t(B):t\geq 0\}$ is a submartingale with respect to the filtration $\mathcal{F}_t$. From the proof of Lemma \ref{lemma:lemmaB}, it follows that $\zeta_t(B)$ is bounded a.s., so by the martingale convergence theorem \cite{Billingsley:2012}, it follows that $\lim_{t\to\infty} \zeta_t(B)$ exists and is finite almost surely. As a result, we have from Lemma \ref{lemma:lemmaB} and (\ref{eq:eq1}):
	\begin{equation*}
		\lim_{t\to\infty} \frac{\zeta_{t+1}(B)}{\zeta_t(B)} \stackrel{a.s.}{=} 1 \stackrel{a.s.}{=} \lim_{t\to\infty} \exp\left( \bv^T \int_B \bD_t(x)\bp_t(x) dx \right)
	\end{equation*}
	Since the variables in the limit on the RHS are bounded a.s., i.e., 
	\begin{align*}
		&\left|v^T\int_B \bD_t(x) \bp_t(x) dx\right| \\
			&\quad \leq \nn v\nn_1 \max_i\left|\int_B [\bD_t(x)]_{i,i} p_{i,t}(x) dx\right| \\
			&\quad \leq \nn v\nn_1 \max_i (2(1-\epsilon_i) \PP_{i,t}(B) - \PP_{i,t}(B)) \\
			&\quad \leq \nn v\nn_1 (1-2 \min_i \epsilon_i) \leq \nn v\nn_1 < \infty,
	\end{align*}
	the dominated convergence theorem for conditional expectations \cite{Durrett:2005} implies, by changing the order of conditional expectation and the limit:
	\begin{equation} \label{eq:eq2}
		\EE\left[ \exp\left( \bv^T \int_B \bD_t(x)\bp_t(x) dx \right) \Bigg| \mathcal{F}_t \right] \stackrel{a.s.}{\longrightarrow} 1
	\end{equation}
	as $t\to\infty$. Substituting the definition of $\bD_t(x)$ into (\ref{eq:eq2}) and using Assumption 1, it follows after some algebra that (\ref{eq:eq2}) is equivalent to:
	\begin{equation} \label{eq:eq3}
		\prod_{i=1}^M \frac{\EE\left[ \exp\left( v_i a_{i,i} \int_B 2l_i(Y_{i,t+1}|x,A_{i,t}) p_{i,t}(x) dx \right) \Bigg| \mathcal{F}_t \right]}{\exp(v_i a_{i,i} \PP_{i,t}(B))} \stackrel{a.s.}{\longrightarrow} 1
	\end{equation}
	Next, we analyze the ratio of exponentials for two separate cases. First, consider the case $\PP_{i,t}([0,b])=\int_{0}^b p_{i,t}(x) dx\leq 1/2$. Using the definition of $\hat{X}_{i,t}$, it follows that $b\leq \hat{X}_{i,t}$. This implies that $l_i(y|x,A_{i,t})=f_1^{(i)}(y)$ for all $x\leq b$. Using this fact and $\PP(Y_{i,t+1}=y|\mathcal{F}_t)=1/2$:
	\begin{align}
		&\frac{\EE\left[ \exp\left( v_i a_{i,i} \int_B 2l_i(Y_{i,t+1}|x,A_{i,t}) p_{i,t}(x) dx \right) \Bigg| \mathcal{F}_t \right]}{\exp(v_i a_{i,i} \PP_{i,t}(B))} \nonumber \\
		&= \frac{1}{2} \frac{\exp\left(v_i a_{i,i} 2(1-\epsilon_i) \PP_{i,t}(B)\right) + \exp\left(v_i a_{i,i} 2 \epsilon_i \PP_{i,t}(B)\right)}{\exp\left(v_i a_{i,i} \PP_{i,t}(B)\right)} \nonumber \\
		&= \frac{1}{2} \Big(\exp\left(v_i a_{i,i} (1-2\epsilon_i) \PP_{i,t}(B)\right)  \nonumber \\
		&\qquad + \exp\left(-v_i a_{i,i} (1-2\epsilon_i) \PP_{i,t}(B)\right) \Big)      \nonumber \\
		&= \cosh\left(v_i a_{i,i} (1-2\epsilon_i) \PP_{i,t}(B) \right) \label{eq:eq3_caseA}
	\end{align}
	where we used the fact that $(e^a+e^{-a})/2=\cosh(a)$. Second, consider the complementary case $\PP_{i,t}([0,b])>1/2$. In this case, we have $b>\hat{X}_{i,t}$ and as a result:
	\begin{align*}
		&\int_0^b 2l_i(Y_{i,t+1}|x,A_{t,i}) p_{i,t}(x) dx \\
			&= \int_0^{\hat{X}_{i,t}} 2f_1^{(i)}(Y_{i,t+1}) p_{i,t}(x) dx + \int_{\hat{X}_{i,t}}^b 2f_0^{(i)}(Y_{i,t+1}) p_{i,t}(x) dx \\
			&= 2 f_1^{(i)}(Y_{i,t+1}) \PP_{i,t}(A_{i,t}) + 2 f_0^{(i)}(Y_{i,t+1}) (\PP_{i,t}(B) - \PP_{i,t}(A_{i,t})) \\
			&= f_1^{(i)}(Y_{i,t+1}) + f_0^{(i)}(Y_{i,t+1}) (2\PP_{i,t}(B)-1) \\
			&= \left\{ \begin{array}{ll} (1-2\epsilon_i)+2\epsilon_i \PP_{i,t}(B), & Y_{i,t+1}=1 \\ 2(1-\epsilon_i) \PP_{i,t}(B) + (2\epsilon_i-1), & Y_{i,t+1}=0 \end{array}  \right.
	\end{align*}
	Using this result and $\PP(Y_{i,t+1}=y|\mathcal{F}_t)=1/2$:
	\begin{align}
		&\frac{\EE\left[ \exp\left( v_i a_{i,i} \int_B 2l_i(Y_{i,t+1}|x,A_{i,t}) p_{i,t}(x) dx \right) \Bigg| \mathcal{F}_t \right]}{\exp(v_i a_{i,i} \PP_{i,t}(B))} \nonumber \\
		&= \frac{1}{2} \frac{1}{\exp\left(v_i a_{i,i} \PP_{i,t}(B)\right)} \times \nonumber \\
		&\quad \Big( \exp\left(v_i a_{i,i} ((1-2\epsilon_i)+2\epsilon_i \PP_{i,t}(B))\right) \nonumber \\
		&\quad +\exp\left(v_i a_{i,i} (2(1-\epsilon_i) \PP_{i,t}(B) + (2\epsilon_i-1)) \right) \Big) \nonumber \\
		&= \frac{1}{2} \Big(\exp\left(v_i a_{i,i} (1-2\epsilon_i) (1-\PP_{i,t}(B))\right) \nonumber \\
		&\qquad + \exp\left(-v_i a_{i,i} (1-2\epsilon_i) (1-\PP_{i,t}(B))\right) \Big) \nonumber \\
		&= \cosh\left(v_i a_{i,i} (1-2\epsilon_i) \PP_{i,t}(B^c) \right) \label{eq:eq3_caseB}
	\end{align}
	Combining the two cases (\ref{eq:eq3_caseA}) and (\ref{eq:eq3_caseB}) by noting that
	\begin{equation*}
		\min\left\{\PP_{i,t}(B),1-\PP_{i,t}(B)\right\}= \left\{ \begin{array}{ll} \PP_{i,t}(B), & \PP_{i,t}(B)\leq 1/2 \\ 1-\PP_{i,t}(B), & \PP_{i,t}(B)>1/2  \end{array} \right. ,
	\end{equation*}
	we have:
	\begin{align*}
		&\frac{\EE\left[ \exp\left( v_i a_{i,i} \int_B 2l_i(Y_{i,t+1}|x,A_{i,t}) p_{i,t}(x) dx \right) \Bigg| \mathcal{F}_t \right]}{\exp(v_i a_{i,i} \PP_{i,t}(B))} \\
		&= \cosh\left(v_i a_{i,i} (1-2\epsilon_i) \min\left\{ \PP_{i,t}(B),1-\PP_{i,t}(B) \right\} \right)
	\end{align*}
  Substituting this expression into (\ref{eq:eq3}), we obtain:
	\begin{equation*} 
		\prod_{i=1}^M  \cosh\left(v_i a_{i,i} (1-2\epsilon_i) \min\left\{\PP_{i,t}(B),1-\PP_{i,t}(B)\right\} \right) \stackrel{a.s.}{\longrightarrow} 1
	\end{equation*}
	Since $1=\cosh(0) \leq \cosh(x)$ for all $x\in \RR$, it follows that $\min\{\PP_{i,t}(B),1-\PP_{i,t}(B)\} \to 0$ almost surely. Note that here we used the positivity of the $v_i$ and the self-reliances $a_{i,i}$ (i.e., Assumption 3) along with the fact that $\epsilon_i < 1/2$. The proof is complete.
\end{proof}

\section{Proof of Lemma \ref{lemma:lemmaD}}
\begin{proof}
Integrating both sides of the recursion (\ref{eq:density_update_matrix}):
\begin{equation} \label{eq:update_B}
	\bP_{t+1}(B) = \bA \bP_t(B) + \bd_{t+1}(B)
\end{equation}
Unrolling (\ref{eq:update_B}) over $r$ steps:
\begin{equation}
	\bP_{t+r}(B) = \bA^r \bP_t(B) + \sum_{k=0}^{r-1} \bA^k \bd_{t+r-k}(B)
\end{equation}
Since $\bA$ is a stochastic matrix, Proposition \ref{prop:Amat} implies:
\begin{align*}
	&V_{t+r}(B) \\
		&= \max_i \PP_{i,t+r}(B) - \min_i \PP_{i,t+r}(B) \\
		&\leq \tau_1(\bA^r) V_t(B) \\
		&\qquad + \max_{i,j} \sum_{k=0}^{r-1} \left( [\bA^k \bd_{t+r-k}(B)]_i - [\bA^k \bd_{t+r-k}(B)]_j \right) \\
		&\leq \tau_1(\bA^r) V_t(B) \\
		&\qquad + \sum_{k=0}^{r-1} \left( \max_i [\bA^k \bd_{t+r-k}(B)]_i - \min_i [\bA^k \bd_{t+r-k}(B)]_i \right) \\
		&\leq \tau_1(\bA^r) V_t(B) \\
		&\qquad + \sum_{k=0}^{r-1} \left( \max_i d_{i,t+r-k}(B) - \min_i d_{i,t+r-k}(B) \right)
\end{align*}
It is known that the coefficient of ergodicity $\tau_1(\bA^r) \in [0,1]$ for any $r\in \NN$ \cite{Ipsen:2011, Seneta:1981}. The irreducibility of the matrix $\bA$ implies the existence of a positive $r$ such that $\tau_1(\bA^r)<1$ \cite{Seneta:1981}.
\end{proof}

\section{Proof of Theorem \ref{thm:thmA}}
\begin{proof}
Without loss of generality, we consider the case $r=1$ in Lemma \ref{lemma:lemmaD}. The case $r>1$ follows similarly. From Lemma \ref{lemma:lemmaD}, we obtain:
\begin{align}
	\EE&[V_{t+1}(B)|\mathcal{F}_t] \leq \tau_1(\bA) V_t(B) \nonumber \\
		&\quad + \EE\left[ \max_i d_{i,t+1}(B) - \min_i d_{i,t+1}(B) \Bigg| \mathcal{F}_t \right]  \label{eq:expA}
\end{align}
where $\tau_1(\bA)<1$. To continue, we need to show that the remainder is asymptotically negligible-i.e.,
\begin{equation*}
	\EE\left[ \max_i d_{i,t+1}(B) - \min_i d_{i,t+1}(B) \Bigg| \mathcal{F}_t \right] \to 0.
\end{equation*}

Using Proposition \ref{prop:lse}, we obtain for any $k>0$:
\begin{align}
	\EE &\left[ \max_i d_{i,t+1}(B) - \min_i d_{i,t+1}(B) \Bigg| \mathcal{F}_t \right] \nonumber \\
	&\leq \frac{1}{k} \EE\Bigg[ \log\left(\sum_{i=1}^M \exp(k d_{i,t+1}(B))\right) \nonumber \\
	&\qquad + \log\left(\sum_{i=1}^M \exp(-k d_{i,t+1}(B))\right) \Bigg| \mathcal{F}_t \Bigg] \nonumber \\
	&\leq \frac{1}{k} \Bigg[ \log\left(\sum_{i=1}^M \EE[\exp(k d_{i,t+1}(B))|\mathcal{F}_t] \right) \nonumber \\
	&\qquad + \log\left(\sum_{i=1}^M \EE[\exp(-k d_{i,t+1}(B))|\mathcal{F}_t] \right) \Bigg] \label{eq:expB}
\end{align}
where we used Jensen's inequality and the linearity of expectation.

Using similar analysis as in the proof of Lemma \ref{lemma:lemmaC}, the (conditional) moment generating functions of the innovation terms can be written as hyperbolic cosines:
\begin{align*}
	\EE[\exp\left(k d_{i,t+1}(B)\right)|\mathcal{F}_t] &= \cosh\left( k a_{i,i} (1-2\epsilon_i) \mu_{i,t}(b) \right) \\
	\EE[\exp\left(-k d_{i,t+1}(B)\right)|\mathcal{F}_t] &= \cosh\left( -k a_{i,i} (1-2\epsilon_i) \mu_{i,t}(b) \right)
\end{align*}
where $\mu_{i,t}(b) = \min\left\{F_{i,t}(b),1-F_{i,t}(b)\right\}$.

Using the even symmetry of the $\cosh(\cdot)$ function, substituting these expressions into (\ref{eq:expB}) and using Proposition \ref{prop:lse} again, we obtain:
\begin{align*}
	&\EE\left[ \max_i d_{i,t+1}(B) - \min_i d_{i,t+1}(B) \Bigg| \mathcal{F}_t \right] \\
		&\leq \frac{2}{k} \log\left(\sum_{i=1}^M \cosh\left( k a_{i,i} (1-2\epsilon_i) \mu_{i,t}(b) \right) \right) \\
		&\leq \frac{2}{k} \log\left(\sum_{i=1}^M \exp\left( k a_{i,i} (1-2\epsilon_i) \mu_{i,t}(b) \right) \right) \\
		&\leq 2 \max_i \left\{ a_{i,i} (1-2\epsilon_i) \mu_{i,t}(b) \right\} + \frac{\log M}{k}
\end{align*}
Taking the limit $k\to\infty$ to tighten the bound and using (\ref{eq:expA}):
\begin{align}
	\EE&[V_{t+1}(B)|\mathcal{F}_t] \leq \tau_1(\bA) V_t(B) + \delta_t  \label{eq:expC}
\end{align}
where $\delta_t := 2 \max_i \left\{ a_{i,i} (1-2\epsilon_i) \mu_{i,t}(b) \right\}$. Lemma \ref{lemma:lemmaC} implies that $\mu_{i,t}(b) \stackrel{a.s.}{\to} 0$, for all $i\in \mathcal{N}$. As a result, we have $\delta_t \stackrel{a.s.}{\to} 0$ as $t\to\infty$.

Taking the unconditional expectation of both sides in (\ref{eq:expC}):
\begin{equation} \label{eq:expD}
	\EE[V_{t+1}(B)] \leq \tau_1(\bA) \EE[V_t(B)] + \EE[\delta_t]
\end{equation}
where $\EE[\delta_t] \to 0$ by the dominated convergence theorem. Using induction on (\ref{eq:expD}), we obtain for all $t\geq 0$:
\begin{equation} \label{eq:expE}
	\EE[V_t(B)] \leq \tau_1(\bA)^t \EE[V_0(B)] + \sum_{l=0}^{t-1} \tau_1(\bA)^l \EE[\delta_{t-1-l}]
\end{equation}
Taking the limits of both sides of (\ref{eq:expE}) and using the fact that $\tau_1(\bA)<1$ and $\EE[V_0(B)]<\infty$:
\begin{align*}
	\limsup_{t\to\infty} &\EE[V_t(B)] \leq \left( \lim_{t\to\infty}\tau_1(\bA)^t \right) \EE[V_0(B)] \nonumber \\
		&+ \lim_{t\to\infty} \sum_{l=0}^{t-1} \tau_1(\bA)^l \EE[\delta_{t-1-l}] = 0
\end{align*}
It follows that $\EE[V_t(B)] \to 0$ since $V_t(B)$ is always nonnegative. Markov's inequality further implies $V_t(B)\stackrel{i.p.}{\to} 0$. The proof is complete.
\end{proof}

\section{Proof of Lemma \ref{lemma:lemmaE}}
\begin{proof}
	From (\ref{eq:density_update}), we evaluate at $x=X^*$ and obtain:
	\begin{align*}
		&p_{i,t+1}(X^*) \\
		  &= a_{i,i} p_{i,t}(X^*) \left( \frac{l_i(Y_{i,t+1}|X^*,A_{i,t})}{\mathcal{Z}_{i,t}(Y_{i,t+1})} \right) + \sum_{j \neq i} a_{i,j} p_{j,t}(X^*) \\
			&= a_{i,i} p_{i,t}(X^*) \left( 2 \PP(Y_{i,t+1}|Z_{i,t}) \right) + \sum_{j \neq i} a_{i,j} p_{j,t}(X^*) \\
	\end{align*}
	where $Z_{i,t}=I(X^*\in A_{i,t})$ is the query input to the noisy channel and $\PP(Y_{i,t+1}|Z_{i,t})$ models the binary symmetric channel for the $i$th agent. Taking the logarithm of both sides and using Jensen's inequality, we obtain for each agent $i$:
	\begin{align*}
		\log p_{i,t+1}(X^*) &\geq \sum_{j=1}^M a_{i,j} \log p_{j,t}(X^*) \\
			&\quad + a_{i,i} \log\left(2 \PP(Y_{i,t+1}|Z_{i,t}) \right) 
	\end{align*}
	Writing this in vector form with the understanding that the logarithm of a vector is taken component-wise:
	\begin{equation} \label{eq:lb_relation}
		\log \bp_{t+1}(X^*) \succeq \bA \log \bp_{t}(X^*) + \diag(\bA) \log \bu_{t+1}
	\end{equation}
	where the vector $\bu_{t+1}$ is given component-wise by $[\bu_{t+1}]_i=2 \PP(Y_{i,t+1}|Z_{i,t})$. Left-multiplying (\ref{eq:lb_relation}) by $\bv^T$ and using the eigenrelation $\bv^T=\bv^T\bA$, we obtain:
	\begin{equation} \label{eq:bnd_recursion}
		\bv^T \log\bp_{t+1}(X^*) \geq \bv^T \log\bp_t(X^*) + \bv^T \diag(\bA) \log \bu_{t+1}
	\end{equation}
	Using induction on (\ref{eq:bnd_recursion}), we obtain:
	\begin{equation*}
		\bv^T \log\bp_t(X^*) \geq \bv^T\log\bp_0(X^*) + \sum_{k=0}^{t-1} \bv^T \diag(\bA) \log\bu_{k+1}
	\end{equation*}
	This implies by the strong law of large numbers (LLN):
	\begin{align*}
		&\liminf_{t\to\infty} \frac{1}{t} \bv^T \log\bp_t(X^*) \\
			&\geq \lim_{t\to\infty} \frac{1}{t} \bv^T \log\bp_0(X^*) + \lim_{t\to\infty} \frac{1}{t} \sum_{k=0}^{t-1} \bv^T \diag(\bA) \log \bu_{k+1} \\
			&= \EE\left[ \sum_{i=1}^M v_i a_{i,i} \log(2 \PP(Y_i|Z_i)) \right] \\
			&= \sum_i v_i a_{i,i} \EE\left[ \log(2 \PP(Y_i|Z_i)) \right]
	\end{align*}
	To finish the proof, note:
	\begin{align*}
		\EE&\left[ \log_2(2\PP(Y_i|Z_i)) \right] \\
			&= \sum_{Z_i} \PP(Z_i) \sum_{Y_i} \PP(Y_i|Z_i) \log_2(2 \PP(Y_i|Z_i)) \\
			&= \sum_{Z_i} \PP(Z_i) \left( (1-\epsilon_i)\log_2(2(1-\epsilon_i)) + \epsilon_i \log_2(2\epsilon_i)  \right) \\
			&= 1- h_B(\epsilon_i) = C(\epsilon_i)
	\end{align*}
\end{proof}

\section{Proof of Theorem \ref{thm:thmB}}
\begin{proof}
	From Theorem \ref{thm:thmA} we obtain for each agent $i$,
	\begin{equation} \label{eq:thm2}
		\PP_{i,t}([0,b]) \stackrel{i.p.}{\longrightarrow} \PP_\infty(B)
	\end{equation}
	as $t\to\infty$, where $\PP_{\infty}(B)$ is a common limiting random variable. To finish the proof, we show that $\PP_\infty(B)$ is the constant $I(b>X^*)$. Lemma \ref{lemma:lemmaE} implies that for $t$ large (as $t\to\infty$)\footnote{The notation $a_n=\Omega(b_n)$ implies $a_n \geq K b_n$ for infinitely many $n$ and for some positive constant $K$.}:
	\begin{equation*}
		\sum_i v_i a_{i,i} \log(p_{i,t}(X^*)) = \Omega(t)
	\end{equation*}
	which implies $\sum_i v_i a_{i,i} \log(p_{i,t}(X^*)) \stackrel{a.s.}{\longrightarrow} +\infty$. This further implies that there exists an agent $i_0$ such that $p_{i_0,t}(X^*) \to\infty$ almost surely.
	
	Lemma \ref{lemma:lemmaC} implies $\mu_{i_0,t}(b') = \min\{F_{i_0,t}(b'), 1-F_{i_0,t}(b')\} \stackrel{a.s.}{\to} 0$ for any $b'\in [0,1]$. This asymptotic result, combined with the monotonicity of the CDF operator $F_{i_0,t}(\cdot)$ and $p_{i_0,t}(X^*) \stackrel{a.s.}{\to} \infty$ imply that $F_{i_0,t}(b) \to I(b>X^*)$. From (\ref{eq:thm2}), it then follows that $F_{i,t}(b) \stackrel{i.p.}{\rightarrow} F_{\infty}(b)=I(b>X^*)$ for all $i \in \mathcal{N}$.
	
	To conclude the proof, we show the conditional mean estimators $\check{X}_{i,t}$ converge to the correct target state $X^*$ in probability (i.e., consistency). From the definition of the conditional expectation, we obtain:
	\begin{align*}
		\check{X}_{i,t} &= \int_{u=0}^1 \PP_{i,t}((u,1]) du \\
			&= 1 - \int_{u=0}^1 F_{i,t}(u) du
	\end{align*}
	where the random variables $F_{i,t}(u)$ are uniformly bounded in $[0,1]$. To finish the proof it suffices to show
	\begin{equation*}
		\int_{u=0}^1 F_{i,t}(u) du \stackrel{i.p.}{\longrightarrow} \int_{u=0}^1 F_{\infty}(u) du
	\end{equation*}
	since $\int_{u=0}^1 F_{\infty}(u) du=1-X^*$. This is accomplished by a variant of the dominated convergence theorem, where the limits are taken in probability. We prove this here for completeness. The first part of the theorem implies
	\begin{equation} \label{eq:convF}
		\limsup_{t\to\infty} \left| F_{i,t}(u) - F_\infty(u)  \right| \stackrel{i.p.}{=} 0
	\end{equation}
	for each $u\in [0,1]\backslash X^*$. Also, we have with probability 1:
	\begin{equation} \label{eq:bndF}
		|F_{i,t}(u)-F_\infty(u)| \leq 2
	\end{equation}
	for all $u\in [0,1]\backslash X^*$ and all $t$. The reverse Fatou lemma along with (\ref{eq:bndF}) and (\ref{eq:convF}) imply:
	\begin{align*}
		\limsup_{t\to\infty}& \int_{0}^1 \left| F_{i,t}(u) - F_\infty(u)  \right| du \\
			&\leq \int_0^1 \limsup_{t\to\infty}\left| F_{i,t}(u) - F_\infty(u)  \right| du \stackrel{i.p.}{=} 0.
	\end{align*}
	Thus, we conclude that:
	\begin{equation*}
		\lim_{t\to\infty} \int_{0}^1 \left| F_{i,t}(u) - F_\infty(u)  \right| du \stackrel{i.p.}{=} 0.
	\end{equation*}
	This concludes the proof.
\end{proof}


%


\bibliographystyle{IEEEtran}
\bibliography{refs}

\begin{thebibliography}{10}
\providecommand{\url}[1]{#1}
\csname url@samestyle\endcsname
\providecommand{\newblock}{\relax}
\providecommand{\bibinfo}[2]{#2}
\providecommand{\BIBentrySTDinterwordspacing}{\spaceskip=0pt\relax}
\providecommand{\BIBentryALTinterwordstretchfactor}{4}
\providecommand{\BIBentryALTinterwordspacing}{\spaceskip=\fontdimen2\font plus
\BIBentryALTinterwordstretchfactor\fontdimen3\font minus
  \fontdimen4\font\relax}
\providecommand{\BIBforeignlanguage}[2]{{%
\expandafter\ifx\csname l@#1\endcsname\relax
\typeout{** WARNING: IEEEtran.bst: No hyphenation pattern has been}%
\typeout{** loaded for the language `#1'. Using the pattern for}%
\typeout{** the default language instead.}%
\else
\language=\csname l@#1\endcsname
\fi
#2}}
\providecommand{\BIBdecl}{\relax}
\BIBdecl

\bibitem{Tsiligkaridis:2013}
T.~Tsiligkaridis, B.~M. Sadler, and A.~O. Hero, ``Collaborative 20 {Q}uestions
  for {T}arget {L}ocalization,'' \emph{IEEE Transactions on Information
  Theory}, vol.~60, no.~4, pp. 2233--2252, April 2014.

\bibitem{Sznitman:2010}
R.~Sznitman and B.~Jedynak, ``Active testing for face detection and
  localization,'' \emph{IEEE Transactions on Pattern Analysis and Machine
  Intelligence}, vol.~32, no.~10, pp. 1914--1920, 2010.

\bibitem{Geman:1996}
D.~Geman and B.~Jedynak, ``An {A}ctive {T}esting model for {T}racking roads in
  {S}atellite {I}mages,'' \emph{IEEE Transactions on Pattern Analysis and
  Machine Intelligence}, vol.~18, no.~1, January 1996.

\bibitem{CastroNowak07}
R.~Castro and R.~Nowak, ``Active learning and sampling,'' in \emph{Foundations
  and Applications of Sensor Management}.\hskip 1em plus 0.5em minus
  0.4em\relax Springer, 2007.

\bibitem{Jadbabaie:2012}
A.~Jadbabaie, P.~Molavi, A.~Sandroni, and A.~Tahbaz-Salehi, ``Non-bayesian
  social learning,'' \emph{Games and Economic Behavior}, vol.~76, pp. 210--225,
  2012.

\bibitem{Jedynak12}
B.~Jedynak, P.~I. Frazier, and R.~Sznitman, ``Twenty questions with noise:
  Bayes optimal policies for entropy loss,'' \emph{Journal of Applied
  Probability}, vol.~49, pp. 114--136, 2012.

\bibitem{Waeber:2013}
R.~Waeber, P.~I. Frazier, and S.~G. Henderson, ``Bisection search with noisy
  responses,'' \emph{SIAM Journal of Control and Optimization}, vol.~53, no.~3,
  pp. 2261--2279, 2013.

\bibitem{Dimakis:2010}
A.~Dimakis, S.~Kar, J.~M.~F. Moura, M.~G. Rabbat, and A.~Scaglione, ``Gossip
  algorithms for distributed signal processing,'' \emph{Proceedings of the
  IEEE}, vol.~98, no.~11, November 2010.

\bibitem{Aysal:2009}
T.~C. Aysal, M.~E. Yildiz, A.~D. Sarwate, and A.~Scaglione, ``Broadcast gossip
  algorithms for consensus,'' \emph{IEEE Transactions on Signal Processing},
  vol.~57, no.~7, July 2009.

\bibitem{Kar:2011}
S.~Kar and J.~M.~F. Moura, ``Covergence rate analysis of distributed gossip
  (linear parameter) estimation: Fundamental limits and tradeoffs,'' \emph{IEEE
  Journal of Selected Topics in Signal Processing}, vol.~5, no.~4, August 2011.

\bibitem{CoverThomas}
T.~D. Cover and J.~A. Thomas, \emph{Elements of Information Theory}.\hskip 1em
  plus 0.5em minus 0.4em\relax Wiley, 2006.

\bibitem{Seneta:1981}
E.~Seneta, \emph{Non-negative Matrices and Markov Chains}, 2nd~ed.\hskip 1em
  plus 0.5em minus 0.4em\relax New York: Springer, 1981.

\bibitem{Ipsen:2011}
I.~C.~F. Ipsen and T.~M. Selee, ``Ergodicity coefficients defined by vector
  norms,'' \emph{SIAM J. Matrix Anal. Appl.}, vol.~32, no.~1, pp. 153--200,
  2011.

\bibitem{Chen:2013}
M.~Chen, S.~C. Liew, Z.~Shao, and C.~Kai, ``Markov approximation for
  combinatorial network optimization,'' \emph{IEEE Transactions on Information
  Theory}, vol.~59, no.~10, October 2013.

\bibitem{Boyd:ConvexOptimization}
S.~Boyd and L.~Vandenberghe, \emph{Convex Optimization}.\hskip 1em plus 0.5em
  minus 0.4em\relax Cambridge University Press, 2004.

\bibitem{Berman:1979}
A.~Berman and R.~J. Plemmons, \emph{Nonnegative matrices in the Mathematical
  Sciences}.\hskip 1em plus 0.5em minus 0.4em\relax Academic Press, New York,
  1979.

\bibitem{Billingsley:2012}
P.~Billingsley, \emph{Probability and Measure}.\hskip 1em plus 0.5em minus
  0.4em\relax John Wiley \& Sons, Inc., Hoboken, NJ, 2012.

\bibitem{Durrett:2005}
R.~Durrett, \emph{Probability: Theory and Examples}, 3rd~ed.\hskip 1em plus
  0.5em minus 0.4em\relax Duxbury Press, Belmont, CA, 2005.

\end{thebibliography}

\end{document}